\newcommand{\extended}[1]{}    \newcommand{\short}[1]{#1}     \renewcommand{\extended}[1]{#1} \renewcommand{\short}[1]{}

\documentclass[sigconf]{aamas}

\usepackage{tikz}
\usetikzlibrary{positioning}
\usetikzlibrary{arrows.meta}
\usepackage{xspace}
\usepackage{float}
\usepackage{subcaption}
\usepackage{xcolor}
\usepackage{multirow}
\usepackage{verbatim}

\usepackage{xspace}
\usepackage{color}
\usepackage{ifthen}
\usepackage{amsmath}

\newcommand{\Actions}{Act}
\newcommand{\Act}{\ensuremath{\mathit{Act}}}

                                    \newcommand{\onepath}[1][]{\ensuremath{\lambda\ifthenelse{\equal{#1}{}}{}{[#1]}}}         \newcommand{\model}{M}

\newcommand{\TransArrow}[1][]{\hookrightarrow....}

\newcommand{\Props}  {\ensuremath{\Pi}}

\newcommand{\lan}[1]{\textsf{#1}\xspace}

\newcommand{\CSL}[1][]{\lan{csl}}
\newcommand{\CSLP}[1][]{\lan{cslp}}
\newcommand{\ATL}[1][]{\lan{ATL$_{\textup{{#1}}}$}}
\newcommand{\ATLs}[1][]{\lan{ATL$_{\textup{{#1}}}$}$\ast$}

\newcommand{\complexityclass}[1]{\ensuremath{\mathbf{{#1}}}\xspace}
\newcommand{\Ptime}{\complexityclass{P}}

\newcommand{\Deltacomplx}[1]{\complexityclass{\Delta_{{#1}}^{\Ptime}}}

\newcommand{\Deltwo}{\Deltacomplx{2}}

\newcommand{\putaway}[1]{}

\newcommand{\para}[1]{\smallskip\noindent\textbf{#1}}

\definecolor{lightgrey}{rgb}{0.8,0.8,0.8}
\definecolor{grey}{rgb}{0.6,0.6,0.6}
\definecolor{darkgrey}{rgb}{0.4,0.4,0.4}
\definecolor{darkgreen}{rgb}{0,0.7,0}

\newenvironment{description2}{\begin{description}\itemsep 0in}{\end{description}}

\newcommand{\finis}{{\scriptsize $\blacksquare$}}
\newcommand{\finisdef}{$\Box$}
\newcommand{\bul}{{\tiny $\blacksquare$}}

\def\itemiremember{\labelitemi}
\def\itemiiremember{\labelitemii}

\newcommand{\coop}[2][]{\langle\!\langle{#2}\rangle\!\rangle_{_{\!\mathit{#1}}}}

\newcommand{\Epath}{\mathsf{E}}
\newcommand{\Apath}{\mathsf{A}}

\newcommand{\Next}[1][]{\!\raisebox{-.2ex}{ \mbox{\unitlength=0.9ex
            \begin{picture}(2,2)
            \linethickness{0.06ex}
            \put(1,1){\circle{2}}   \end{picture}}}_{{#1}}  \,}
\newcommand{\Sometm}[1][]{\Diamond_{{#1}}}
\newcommand{\Always}[1][]{\Box_{{#1}}}
\newcommand{\Until}[1][]{\,\mathcal{U}_{{#1}}\,}

\newcommand{\plaus}[1][]{\ifthenelse{\equal{#1}{}}{\mathbf{P\;\!\!l}\,}{\mathbf{P\;\!\!l}_{#1}\,}}
\newcommand{\phys}[1][]{\ifthenelse{\equal{#1}{}}{\mathbf{P\;\!\!h}\,}{\mathbf{P\;\!\!h}_{#1}\,}}

\newcommand{\plaumodels}[1][]{\ensuremath{\ifthenelse{\equal{#1}{}}{\models_\sPlaupaths}{\models_{#1}}}}
\newcommand{\sPlaupaths}{\ensuremath{P}}

\newcommand{\then}{\rightarrow}

\newcommand{\true}{\top}
\newcommand{\false}{\bot}

\newcommand{\set}[1]{\{{#1}\}}

\newcommand{\tuple}[1]{\langle{#1}\rangle}

\newcommand{\Agt}{\ensuremath{\mathbb{A}\mathrm{gt}}}

\newcommand{\prop}[1]{\ensuremath{\mathsf{{#1}}}}

\newcommand{\act}{\mathsf{act}}

\usepackage{balance}

\setcopyright{ifaamas}
\acmConference[AAMAS '23]{Proc.\@ of the 22nd International Conference
on Autonomous Agents and Multiagent Systems (AAMAS 2023)}{May 29 -- June 2, 2023}
{London, United Kingdom}{A.~Ricci, W.~Yeoh, N.~Agmon, B.~An (eds.)}
\copyrightyear{2023}
\acmYear{2023}
\acmDOI{}
\acmPrice{}
\acmISBN{}

\acmSubmissionID{877}

\title[Playing to Learn, or to Keep Secret]{Playing to Learn, or to Keep Secret:\\ Alternating-Time Logic Meets Information Theory}

\author{Masoud Tabatabaei}
\affiliation{
  \institution{Interdisciplinary Centre for Security and Trust, SnT, \\ University of Luxembourg}
  \city{}
  \country{}}
\email{masoud.tabatabaei@uni.lu}

\author{Wojciech Jamroga}
\affiliation{
  \institution{Institute of Computer Science, Polish Academy of Sciences}
  \city{and SnT, University of Luxembourg}
  \country{}}
\email{wojciech.jamroga@uni.lu}

\begin{abstract}
Many important properties of multi-agent systems refer to the participants' ability to achieve a given goal, or to prevent the system from an undesirable event. Among intelligent agents, the goals are often of epistemic nature, i.e., concern the ability to obtain knowledge about an important fact $\varphi$. Such properties can be e.g.~expressed in \ATLK, that is, alternating-time temporal logic \ATL extended with epistemic operators.
In many realistic scenarios, however, players do not need to fully learn the truth value of $\varphi$. They may be almost as well off by gaining \emph{some} knowledge; in other words, by reducing their uncertainty about $\varphi$. Similarly, in order to keep $\varphi$ secret, it is often insufficient that the intruder never fully learns its truth value. Instead, one needs to require that his uncertainty about $\varphi$ never drops below a reasonable threshold.

With this motivation in mind, we introduce the logic \ATLH, extending \ATL with quantitative modalities based on the Hartley measure of uncertainty. The new logic enables to specify agents' abilities w.r.t. the uncertainty of a given player about a given set of statements. It turns out that \ATLH has the same expressivity and model checking complexity as \ATLK. However, the new logic is exponentially more succinct than \ATLK, which is the main technical result of this paper.
\end{abstract}


\newcommand{\pairModels}[2]{\langle #1 \circ #2 \rangle}
\renewcommand{\ATLs}{\textbf{ATL*}\xspace}
\newcommand{\CTL}{\textbf{CTL}\xspace}

\renewcommand{\ATL}{\textbf{ATL}\xspace}
\newcommand{\ATLK}{\textbf{ATLK}\xspace}
\renewcommand{\Next}{X}
\renewcommand{\Always}{G}
\renewcommand{\Sometm}{F}

\newcommand{\notA}{\overline{A}}
\newcommand{\notB}{\overline{B}}

\newtheorem*{theorem*}{Theorem}
\newtheorem{observation}{Observation}

\newcommand{\States}{\ensuremath{St}}
\renewcommand{\Agt}{\mathit{Agt}}
\renewcommand{\Act}{\mathit{Act}}
\renewcommand{\act}{\mathit{d}}
\newcommand{\outcome}{\mathit{out}}
\newcommand{\indist}{\sim}
\newcommand{\atomicP}{\mathit{Pv}}
\renewcommand{\Props}{\mathit{Prop}}
\newcommand{\valualtion}{\mathit{V}}
\newcommand{\uncertain}{\mathcal{H}}
\newcommand{\knows}{\mathcal{K}}
\newcommand{\ATHL}{\textbf{ATLH}\xspace}
\newcommand{\ATLH}{\textbf{ATLH}\xspace}
\newcommand{\ATEL}{\textbf{ATLK}\xspace}
\newcommand{\SL}{\textbf{SL}\xspace}

\newcommand{\sOne}{$\mathbf{S1_j}$\xspace}
\newcommand{\sTwo}{$\mathbf{S2_j}$\xspace}
\newcommand{\sThree}{$\mathbf{S3_j}$\xspace}
\newcommand{\sFour}{$\mathbf{S4_j}$\xspace}
\newcommand{\sFive}{$\mathbf{S5_j}$\xspace}
\newcommand{\sSix}{$\mathbf{S6_j}$\xspace}

\renewcommand{\act}{\mathit{act}}

\newcommand{\highlightone}[1]{\underline{#1}}
\newcommand{\highlighttwo}[1]{\underline{#1}}

\begin{document}

\pagestyle{fancy}
\fancyhead{}

\maketitle

\section{Introduction}\label{sec:introduction}

Many important properties of multi-agent systems refer to \emph{strategic abilities} of agents and their groups~\cite{Goranko15stratmas,Agotnes15handbook}. They can be formalized in logics of strategic ability, such as alternating-time temporal logic \ATL~\cite{Alur02ATL,Schobbens04ATL} or strategy logic \SL~\cite{Chatterjee07strategylogic,Mogavero10stratLogic}.
For example, the \ATL formula $\coop{taxi}\Sometm\prop{destination}$, built upon the strategic operator $\coop{A}\varphi$ for ``there is a strategy for $A$ such that $\varphi$ holds'' and the temporal modality $\Sometm$ (``eventually''), can be used to express that the autonomous cab can deliver the passenger to his/her destination.
Similarly, $\coop{passg}(\lnot\prop{dead})\Until\prop{destination}$ says that the passenger has the ability to survive the ride alive.
Such statements allow to express important functionality and safety requirements in a simple and intuitive way.
Moreover, they provide input to algorithms and tools for verification, that have been in constant development for over 20 years~\cite{Alur98mocha-cav,Alur01jmocha,Kacprzak04umc-atl,Lomuscio06mcmas,Chen13prismgames,Busard14improving,Pilecki14synthesis,Huang14symbolic-epist,Cermak14mcheckSL,Lomuscio17mcmas,Cermak15mcmas-sl-one-goal,Belardinelli17broadcasting,Belardinelli17broadcasting,Jamroga19fixpApprox-aij,Kurpiewski19stv-demo,Kurpiewski21stv-demo}.

Knowledge and information has always been an important aspect of interaction, but it has become even more important with the emergence of Internet and, more recently, social networks.
Information is an important resource on which strategies are built, e.g., it is widely acknowledged that executable strategies must comply with so called uniformity constraints~\cite{Schobbens04ATL,Jamroga04ATEL}.
More and more often, information becomes also the \emph{goal} of the interaction. Agents may play to \emph{learn} about a particular subject. People strive to know what the state of the economy is, what is the latest clothing fashion, and whether the coffee machine at their workplace serves good espresso or not. Using strategic-epistemic specifications that involve the knowledge operator $K_a$, the latter kind of ability can be expressed by $\coop{worker}\Sometm(K_{worker}\prop{good}\lor K_{worker}\neg\prop{good})$.
Dually, the user of a social network may want to post a message for their friends only, in which case no outsider should learn the content of the message.
This kind of ability can be captured by $\coop{user}\Always\neg(K_{outsider}\prop{post=m}\lor K_{outsider}\prop{post\neq m})$.

In many cases, however, strategic-epistemic specifications are too coarse. It is great if the worker can obtain full knowledge about the quality of workplace espresso, but being \emph{almost sure} is almost as good.
Dually, leaking \emph{some} information about the social network post can be damaging, even if the intruder does not learn its exact content.
With this motivation in mind, we propose to extend alternating-time temporal logic with new, information-theoretic modalities $\uncertain$, based on the Hartley measure of uncertainty~\cite{Hartley28information}.
We also demonstrate the usefulness of the framework on a real-life voting scenario.

In terms of technical results, we prove that the resulting logic has the same expressive power and model checking complexity as strategic-epistemic specifications; however, it is exponentially more succinct.
This is an important result, as it shows that the verification of a given \emph{property} with uncertainty operators can take exponentially less time than when one uses knowledge modalities.

\para{Related work.} Strategic-epistemic reasoning has been intensively studied in the early 2000s, especially within the framework of ATEL~\cite{Hoek02ATEL,Hoek03ATELstudialogica,Hoek03ATELcasestudy,Agotnes06action,Jamroga07constructive-jancl} and Dynamic Epistemic Logic~\cite{Ditmarsch07DEL,Agotnes08coalitions}. Dynamic epistemic planning~\cite{Bolander11epistplanning} is a particularly relevant example.
Still, we are not aware of any works combining logical formalizations of strategic reasoning with information-theoretic properties.
The paper~\cite{Jamroga13accumulative} comes closest, as it discusses the relation between a variant of resource-bounded temporal-epistemic logic and Hartley measure.
Moreover, our proposal is directly inspired by information-theoretic notions of security, cf.~\cite{KatzLindell20crypto} for an introduction.

Another strand of related works concerns quantitative specification and verification of MAS due to stochastic interaction~\cite{Chen07PATL,Huang12probabilisticATL}, graded~\cite{Aminof18gradedSL,Ferrante08enriched-mu} and fuzzy strategic modalities~\cite{Bouyer19fuzzySL,Belardinelli22humanfriendly}, or probabilistic beliefs about the opponents' response~\cite{Bulling09patl-fundamenta}. Those papers considered neither knowledge nor information-theoretic properties, though~\cite{Ferrante08enriched-mu} leaned in that direction by including a count over the accessible imperfect worlds.

Succinctness of logical representations has been studied since early 1970s~\cite{Stockmeyer72phd}.
In particular, the relative succinctness of branching-time logics was investigated in~\cite{Wilke99succinctness,Adler01succinctness,Markey03succinctness}, and~\cite{Bozzelli20ATLwithPast} studied the succinctness of the strategic logic \ATLs with past-time operators.
The methodology of proving succinctness by means of \emph{formula size games} was proposed in~\cite{Adler01succinctness}, and later generalized in~\cite{French13succinctness}. We adapt the latter approach to obtain our main technical result here.

\section{Logics of Strategic Ability}\label{expressing-sec-logics}

We first recapitulate the logical foundations that we chose for our approach.

\subsection{Alternating-Time Logic \ATL}\label{sec:atl}

\emph{Alternating-time temporal logic} \ATL~\cite{Alur97ATL,Alur02ATL,Schobbens04ATL} generalizes the branching-time temporal logic \CTL~\cite{Emerson90temporal} by replacing the path quantifiers $\Epath,\Apath$ with \emph{strategic modalities} $\coop{A}$.
Informally, $\coop{A}\gamma$ says that a group of agents $A$ has a collective strategy to enforce temporal property $\gamma$.
\ATL formulas can include temporal operators: ``$\Next$'' (``in the next state''), ``$\Always$'' (``always from now on''), ``$\Sometm$'' (``now or sometime in the future''),
and $\Until$ (strong ``until'').

\para{Syntax.}
Formally, let $\Agt$ be a finite set of agents, and $\Props$ a countable set of atomic propositions.
The language of \ATL is defined as follows:
\begin{center}
  $\varphi::= \prop{p} \mid \neg \varphi \mid \varphi\wedge\varphi
    \mid \coop{A}\Next\varphi \mid \coop{A}\Always\varphi \mid \coop{A}\varphi\Until\varphi$.
\end{center}
where $A\subseteq\Agt$ and $\prop{p} \in \Props$.
Derived boolean connectives and constants ($\lor,\true,\false$) are defined as usual.
``Sometime'' is  defined as $\Sometm\varphi \equiv \true \Until \varphi$.

\subsection{Semantics of \ATL}

\para{Models.}
The semantics of \ATL is defined over a variant of synchronous multi-agent transition systems.
Let $S = \set{\Agt, \States, \Act, d, \outcome}$ be a concurrent game structure (CGS) such that:
 $\Agt = \set{1,...,k}$ is a set of agents (or players),
 $\States$ is the set of states of the system,
 $\act$ is a set of actions,
 $d: \Agt\times\States \then 2^{\Act}\setminus\set{\emptyset}$ shows what actions are available for each player in each state,
 and $\outcome: \States \times \Act^1 \times ... \times \Act^k \then \States$ is the transition function which, given a state and one action from each player in that state, returns the resulting sate.
 A CGS together with a set of atomic propositions $\atomicP$ and a valuation function $\valualtion: \atomicP \then 2^{\States}$ is called a concurrent game model (CGM).
 A \emph{pointed CGM} is a pair $(M,q_0)$ consisting of a concurrent game model $M$ and an initial state $q_0$ in $M$.

\para{Strategies and their outcomes.}
Given a CGS, we define the strategies and their outcomes as follows.
A \emph{strategy} for $a\in\Agt$ is a function $s_a:\States\rightarrow \Act$ such that $s_a(q)\in d(a,q)$.\footnote{
  This corresponds to the notion of \emph{memoryless} or \emph{positional} strategies. In other words, we assume that the memory of agents is explicitly defined by the states of the model. }
The set of such strategies is sometimes denoted by $\Sigma_a^{Ir}$, with the capital ``I'' referring to perfect \textbf{I}nformation, and the lowercase ``r'' for possibly imperfect \textbf{r}ecall.
A \emph{collective strategy} for a group of agents $A=\set{a_1,\ldots,a_r}$
is a tuple of individual strategies $s_A = \tuple{s_{a_1}, \ldots, s_{a_r}}$.
The set of such strategies is denoted by $\Sigma_A^{Ir}$.

A \emph{path} $\lambda=q_0q_1q_2\dots$ in a CGS is an infinite sequence of states
such that there is a transition between each $q_i$ and $q_{i+1}$.
$\lambda[i]$ denotes the $i$th position on $\lambda$ (starting from $i=0$) and $\lambda[i,\infty]$ the suffix of $\lambda$ starting with $i$.
The ``outcome'' function $out(q,s_A)$ returns the set of all paths that
can occur when agents $A$ execute strategy $s_A$ from state $q$ onward, defined as follows:

\begin{description}
  \item[$out(q,s_A) =$]  $\{ \lambda=q_0,q_1,q_2\ldots \mid
      q_0=q$ and for each $i=0,1,\ldots$ there exists
    $\tuple{\alpha^{i}_{a_1},\ldots,\alpha^{i}_{a_k}}$ such that
    $\alpha^{i}_{a} \in d_a(q_{i})$ for every $a\in \Agt$,
    and $\alpha^{i}_{a} = s_A[a](q_{i})$ for every $a\in A$,
    and $q_{i+1} = o(q_{i},\alpha^{i}_{a_1},\ldots,\alpha^{i}_{a_k}) \}$.
\end{description}

\para{Semantic clauses.}
The semantics of \ATL is defined by the following clauses:
\begin{description2}
  \item[$\model,q\state \models \prop{p}$] iff  $q \in V(\prop{p})$, for $\prop{p}\in\Props$;
  \item[$\model,q\state \models \neg\varphi$] iff $\model,\state \not\models  \varphi$;
  \item[$\model,q\state \models \varphi_1\wedge\varphi_2$] iff $\model,q \models \varphi_1$ and $\model,q \models \varphi_2$;
\item[$\model,q \models \coop{A}\Next\varphi$] iff there is a strategy $s_A\in\Sigma_A^{Ir}$
  such that, for each path $\lambda\in out(q,s_A)$, we have $\model,\lambda[1] \models \varphi$.
  \item[$\model,q \models \coop{A}\Always\varphi$] iff there is a strategy $s_A\in\Sigma_A^{Ir}$
  such that, for each path $\lambda\in out(q,s_A)$, we have $\model,\lambda[i] \models \varphi$ for all $i\geq 0$.
  \item[$\model,q \models \coop{A}\varphi_1\Until\varphi_2$] iff there is a strategy $s_A\in\Sigma_A^{Ir}$
  such that, for each path $\lambda\in out(q,s_A)$, we have
  $\model,\lambda[i] \models \varphi_2$ for some $i\geq 0$ and $\model, \lambda[j] \models \varphi_1$
  for all $0\leq j< i$.
\end{description2}

\subsection{Imperfect Information and Knowledge}\label{sec:imperfect}

Realistic multi-agent interaction always includes some degree of limited observability. Here, we use the classical variant of ``\ATL with imperfect information'', defined as follows:

We extend concurrent game structures with indistinguishability relations $\indist_a$, for each $a\in\Agt$.
The resulting structure $S = \set{\Agt, \States, \set{\indist_a \mid a\in\Agt}, \Act, d, \outcome}$ is called a concurrent epistemic game structure (CEGS).
A CEGS together with a set of atomic propositions $\atomicP$ and a valuation function $\valualtion: \atomicP \then 2^{\States}$ is called a concurrent epistemic game model CEGM.

Strategies of agents must specify identical choices in indistinguishable situations. That is,
 strategies with imperfect information (\textit{ir} strategies, for short) are functions $s_a : \States\to\Actions$ such that (1) $s_a(q)\in d(a,q)$, and (2) if $q\sim_a q'$ then $s_a(q)=s_a(q')$.
 \footnote{Again, we consider only positional strategies here.}
As before, collective strategies for $A\subseteq\Agt$ are tuples of individual strategies for $a\in A$. We denote the set of $A$'s imperfect information strategies by $\Sigma_A^{ir}$ (with the lowercase ``i'' for imperfect \textbf{i}nformation).

The semantics of ``\ATL with imperfect information'' ($\textbf{ATL}_{ir}$) differs from the one presented in Section~\ref{sec:atl} only in that the strategies are taken from $\Sigma_A^{ir}$ instead of $\Sigma_A^{Ir}$.
 In other words, the agents in $A$ should have an executable strategy which enforces $\varphi$ from all the states that at least one member of the coalition considers possible.

\textit{Alternating-time temporal epistemic logic} \ATEL adds the knowledge modality of the \textit{multi-agent epistemic logic} to \ATL with imperfect information.
 In multi-agent epistemic logic, expressing the \emph{knowledge} of the agents is formalised by epistemic formulae of type $\knows_a \varphi$, stating that agent $a$ \emph{knows that $\varphi$ holds}, with the following semantics:
\begin{description2}
  \item $\model,q \models \knows_a \varphi$\quad iff, for every state $q'$ such that $q\sim_a q'$, we have that $\model,q' \models \varphi$.
\end{description2}

The formula stating mutual knowledge $E_A\varphi$ (read ''everybody in $A$ knows that $\varphi$'') is defined as:
\begin{center}
  $M,q\models E_A\varphi$ iff $M,q\models \knows_a \varphi$, for all $a\in A$.
\end{center}

\section{Motivating Example}\label{sec:motivating}

In this section we show, using an example, how our proposed logic can express more refined epistemic goals of agents using considerably more concise formulas.
As we will see, not only the new formulation of these epistemic properties will be significantly shorter; the interpretation of the formulas will also be easier to understand in comparison to their analogous formulas in \ATLK.

\subsection{Coercion in Referendums}

We consider a very simple scenario of an election with a single voter and a single coercer. The election is a referendum, in the sense that each voter has to either vote for an issue in question
or to vote against it. We consider two variants. In the first one there is only one issue put for referendum (we call it proposal \textit{A}). The model consists of two agents, the voter $v$ and the coercer $c$.
The set of possible actions for the coercer in the model is $\set{\epsilon}$ and for the voter is $\set{voteA, vote\overline{A}, \epsilon}$, where $\epsilon$ represents a \textit{null}
action (meaning the action of doing nothing). $voteA$ and $vote\overline{A}$ respectively represent voting for and against the proposal $A$. The atomic proposition $V_A$
states that the vote is cast in favor of proposal $A$, while the atomic proposition $Voted$ shows that the vote is already cast. The game model is depicted in Figure~\ref{fig:voterExample1}.

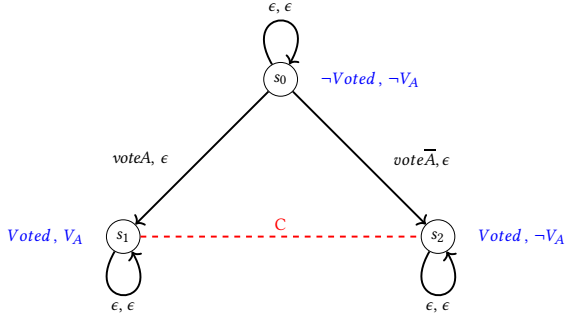
\begin{figure}[t]\centering
\resizebox{0.9\columnwidth}{!}{\begin{tikzpicture}[node distance = 4cm]
      \node[draw, shape = circle] (s0) at (0,0) {$s_0$} ;
      \node[draw, shape = circle] (s1) [below left of= s0] {$s_1$};
      \node[draw, shape = circle] (s2) [below right of= s0] {$s_2$};

      \node[right = 0.3cm of s0, blue] {$\lnot Voted\,, \, \lnot V_A$};
      \node[left = 0.3cm of s1, blue] {$Voted\,, \, V_A$};
      \node[right = 0.3cm of s2, blue] {$Voted\,, \, \lnot V_A$};

      \path[->, line width = 1 ] (s0) edge
      node[left, text width = 1.5cm] {\textit{voteA}, $\epsilon$}
      (s1);

      \path[->, line width = 1 ] (s0) edge [loop above, looseness=10,out=120, in = 60]
      node[above] {$\epsilon$, $\epsilon$}
      (s0);

      \path[->, line width = 1 ] (s0) edge
      node[right = 0.5cm, text width = 2cm] {$  vote\overline{A} , \epsilon$}
      (s2);

      \path[dashed, line width = 1, red ] (s1) edge
      node[above] {C}
      (s2);

      \path[->, line width = 1 ] (s1) edge [loop above, looseness=10,out=240, in = 300]
      node[below] {$\epsilon$, $\epsilon$}
      (s1);

      \path[->, line width = 1 ] (s2) edge [loop above, looseness=10,out=240, in = 300]
      node[below] {$\epsilon$, $\epsilon$}
      (s2);
    \end{tikzpicture}  }
\caption{Single issue referendum with one voter and one coercer}\label{fig:voterExample1}
\end{figure}

The valuations of atomic propositions are depicted in blue, and the red dashed line between $s_1$ and $s_2$ shows that the these two states are indistinguishable for
player $c$. In this simple model, we might express the property of coercion resistance in \ATLK as follows:
  \begin{align*}
    M, s_0 & \models \coop{v}\Sometm(Voted \wedge V_A \wedge \Always\lnot(\knows_c V_A \vee \knows_c \lnot V_A)) \\
    \wedge & \coop{v}\Sometm(Voted \wedge \lnot V_A \wedge \Always\lnot(\knows_c V_A \vee \knows_c \lnot V_A))
  \end{align*}

The formula states that the voter has a strategy to vote for the proposal $A$ or against it, in a way that in either case the coercer does not know the value of the vote.

\subsection{Referendum with Multiple Proposals}

Consider a more sophisticated variant in which the voter participates in a double referendum, i.e., votes on two proposals $A$ and $B$ on a single ballot.
The set of atomic propositions in this scenario is $\set{Voted, V_A, V_B}$ and the set of actions for the voter is\linebreak
$\set{voteAB, vote\overline{AB}, voteA\overline{B}, vote\overline{B}A}$. For expressing the property of coercion resistance in this scenario, a seemingly reasonable way
is to extend the above formulas as below:
  \begin{align*}
    M, s_0 & \models \coop{v}\Sometm(Voted \wedge V_A \wedge V_B \wedge                                       \\
           & \qquad \Always\lnot(\knows_c V_A \vee \knows_c \lnot V_A \vee \knows_c V_B \vee \knows_c \lnot V_B)) \\
           & \wedge \coop{v}\Sometm(Voted \wedge V_A \wedge \lnot V_B \wedge                                  \\
           & \qquad \Always\lnot(\knows_c V_A \vee \knows_c \lnot V_A \vee \knows_c V_B \vee \knows_c \lnot V_B)) \\
           & \wedge \coop{v}\Sometm(Voted \wedge \lnot V_A \wedge V_B \wedge                                  \\
           & \qquad \Always\lnot(\knows_c V_A \vee \knows_c \lnot V_A \vee \knows_c V_B \vee \knows_c \lnot V_B)) \\
           & \wedge \coop{v}\Sometm(Voted \wedge \lnot V_A \wedge \lnot V_B \wedge                            \\
           & \qquad \Always\lnot(\knows_c V_A \vee \knows_c \lnot V_A \vee \knows_c V_B \vee \knows_c \lnot V_B))
  \end{align*}

The formula states that the voter can vote in any combination, for or against $A$ or $B$, without the coercer knowing the value of any single vote. At the first glance these security properties
seem to be strong enough for capturing the desirable property of coercion resistance. However, if we look at the two models in Figure~\ref{fig:voterExample2}, both of
them satisfy the property above. On the other hand, we would consider model $M_1$ less secure than $M_2$.
There are 4 possible combinations of the valuations of $V_A$ and $V_B$. In $M_2$, the coercer considers all 4 of them as plausible, but in $M_1$ he could narrow that down to only two possible combinations.
In other words, the uncertainty of the coercer about propositions $V_A$ and $V_B$ is higher in $M_2$ than in $M_1$. In fact, as we shall see later, it is possible
to write a formula in the language of $\ATLK$ that keeps the above property and yet distinguishes $M_1$ and $M_2$. But if we want to write a security property in \ATLK that rejects all the models where
the coercer has more distinguishing power over states $s_1$ to $s_4$ than the model $M_2$, then the length of that formula would be very large -- in the worst case, even exponential in the number of distinguishing properties.

\begin{figure}[t]
  \begin{subfigure}[b]{0.5\textwidth}
    \scalebox{0.7} {
      \begin{tikzpicture}[node distance = 3cm]
        \node[draw, shape = circle] (s0) at (0,0) {$s_0$} ;
        \node[draw, shape = circle] (s2) [below left of= s0] {$s_2$};
        \node[draw, shape = circle] (s3) [below right of= s0] {$s_3$};
        \node[draw, shape = circle] (s1) [left of= s2] {$s_1$};
        \node[draw, shape = circle] (s4) [right of= s3] {$s_4$};

        \node[above right = -0.3cm and 0.3cm of s0, blue] {$\lnot Voted\,, \, \lnot V_A\, , \, \lnot V_B$};
        \node[above = 0.3cm of s1, blue] {$Voted\,, \, V_A\, , \, \lnot V_B$};
        \node[below right = 0.0 cm and 0.1cm of s2, blue] {$Voted\,, \, \lnot V_A\, , \, V_B$};
        \node[below right = 0.0 cm and 0.1cm of s3, blue] {$Voted\,, \, V_A\, , \, V_B$};
        \node[above = 0.3cm of s4, blue] {$Voted\,, \, \lnot V_A\, , \, \lnot V_B$};

        \path[->, line width = 1 ] (s0) edge [loop above, looseness=10,out=120, in = 60]
        node[above] {$\epsilon$, $\epsilon$}
        (s0);

        \path[->, line width = 1 ] (s0) edge
        node[above left = 0.2cm and -0.5cm, text width = 1.5cm] {$voteA\overline{B}\, ,\, \epsilon$}
        (s1);

        \path[->, line width = 1 ] (s0) edge
        node[below right = 0.2cm and -0.4cm, text width = 1.5cm] {$vote\overline{A}B\, ,\, \epsilon$}
        (s2);

        \path[->, line width = 1 ] (s0) edge
        node[above left = -0.2cm and 0.15cm] {$voteAB\, ,\, \epsilon$}
        (s3);

        \path[->, line width = 1 ] (s0) edge
        node[above right = 0.2 and -0.3cm] {$vote\overline{AB}\, ,\, \epsilon$}
        (s4);

        \path[dashed, line width = 1, red ] (s1) edge
        node[above] {C}
        (s2);

        \path[dashed, line width = 1, red ] (s3) edge
        node[above] {C}
        (s4);

        \path[->, line width = 1 ] (s1) edge [loop above, looseness=10,out=240, in = 300]
        node[below] {$\epsilon$, $\epsilon$}
        (s1);

        \path[->, line width = 1 ] (s2) edge [loop above, looseness=10,out=240, in = 300]
        node[below] {$\epsilon$, $\epsilon$}
        (s2);

        \path[->, line width = 1 ] (s3) edge [loop above, looseness=10,out=240, in = 300]
        node[below] {$\epsilon$, $\epsilon$}
        (s3);

        \path[->, line width = 1 ] (s4) edge [loop above, looseness=10,out=240, in = 300]
        node[below] {$\epsilon$, $\epsilon$}
        (s4);
      \end{tikzpicture}
    }
    \caption{$M_1$} \label{fig:M1}
  \end{subfigure}

  \vspace{0.5cm}

  \begin{subfigure}[b]{0.5\textwidth}
    \scalebox{0.7} {
      \begin{tikzpicture}[node distance = 3cm]
        \node[draw, shape = circle] (s0) at (0,0) {$s_0$} ;
        \node[draw, shape = circle] (s2) [below left of= s0] {$s_2$};
        \node[draw, shape = circle] (s3) [below right of= s0] {$s_3$};
        \node[draw, shape = circle] (s1) [left of= s2] {$s_1$};
        \node[draw, shape = circle] (s4) [right of= s3] {$s_4$};

        \node[above right = -0.3cm and 0.3cm of s0, blue] {$\lnot Voted\,, \, \lnot V_A\, , \, \lnot V_B$};
        \node[above = 0.3cm of s1, blue] {$Voted\,, \, V_A\, , \, \lnot V_B$};
        \node[below right = 0.0 cm and 0.1cm of s2, blue] {$Voted\,, \, \lnot V_A\, , \, V_B$};
        \node[below right = 0.0 cm and 0.1cm of s3, blue] {$Voted\,, \, V_A\, , \, V_B$};
        \node[above = 0.3cm of s4, blue] {$Voted\,, \, \lnot V_A\, , \, \lnot V_B$};

        \path[->, line width = 1 ] (s0) edge [loop above, looseness=10,out=120, in = 60]
        node[above] {$\epsilon$, $\epsilon$}
        (s0);

        \path[->, line width = 1 ] (s0) edge
        node[above left = 0.2cm and -0.5cm, text width = 1.5cm] {$voteA\overline{B}\, ,\, \epsilon$}
        (s1);

        \path[->, line width = 1 ] (s0) edge
        node[below right = 0.2cm and -0.4cm, text width = 1.5cm] {$vote\overline{A}B\, ,\, \epsilon$}
        (s2);

        \path[->, line width = 1 ] (s0) edge
        node[above left = -0.2cm and 0.15cm] {$voteAB\, ,\, \epsilon$}
        (s3);

        \path[->, line width = 1 ] (s0) edge
        node[above right = 0.2 and -0.3cm] {$vote\overline{AB}\, ,\, \epsilon$}
        (s4);

        \path[dashed, line width = 1, red ] (s1) edge
        node[above] {C}
        (s2);

        \path[dashed, line width = 1, red ] (s2) edge
        node[above] {C}
        (s3);

        \path[dashed, line width = 1, red ] (s3) edge
        node[above] {C}
        (s4);

        \path[->, line width = 1 ] (s1) edge [loop above, looseness=10,out=240, in = 300]
        node[below] {$\epsilon$, $\epsilon$}
        (s1);

        \path[->, line width = 1 ] (s2) edge [loop above, looseness=10,out=240, in = 300]
        node[below] {$\epsilon$, $\epsilon$}
        (s2);

        \path[->, line width = 1 ] (s3) edge [loop above, looseness=10,out=240, in = 300]
        node[below] {$\epsilon$, $\epsilon$}
        (s3);

        \path[->, line width = 1 ] (s4) edge [loop above, looseness=10,out=240, in = 300]
        node[below] {$\epsilon$, $\epsilon$}
        (s4);
      \end{tikzpicture}  }
    \caption{$M_2$} \label{fig:M2}
  \end{subfigure}
  \caption{Double referendum with one voter and one coercer. Model $M_1$ depicts a scenario which is less secure than $M2$}
  \label{fig:voterExample2}
\end{figure}
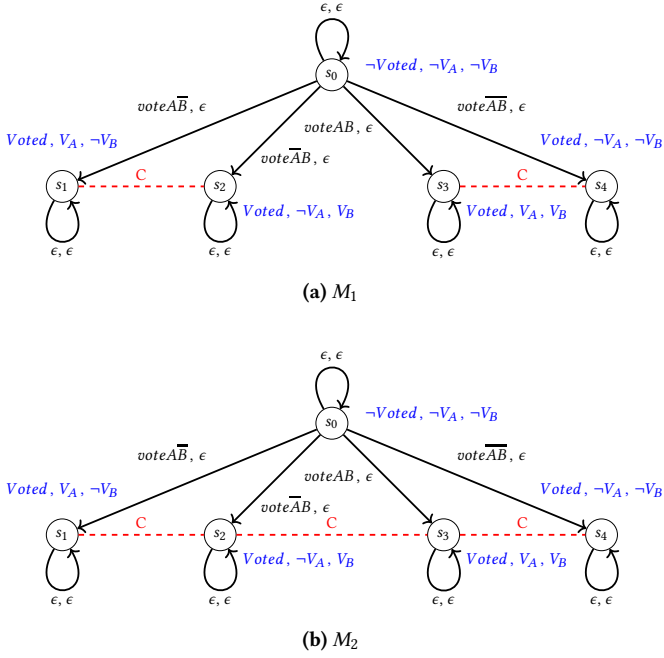

\subsection{Reasoning about Uncertainty}\label{sec:uncertainty}

One way of looking at the above situation is that, when reaching any of the states $s_1$ to $s_4$, we want the coercer to have the least amount of information, or in other words the maximum uncertainty
about the possible values of $V_A$ and $V_B$. To express this concept, we can use one of the well known quantitative measures of uncertainty. Two measures that come to mind are Shannon entropy and Hartley measure.
Choosing Shannon entropy would be meaningful only if we knew the intrinsic probabilities of each state. However, in the models that we are using, and in the scenarios similar
to the one above, what we are interested is the uncertainty of the agents about different possible outcomes of a set of properties (here $V_A$ and $V_B$).
We recall the definition of Hartley measure below:

\begin{definition}[Hartley measure of uncertainty~\cite{Hartley28information}]\label{def:hartley}
  If $A$ is a set of possible outcomes, then its Hartley measure is defined by $H(A) = \log |A|$.
\end{definition}

The Hartley function coincides with Shannon entropy when ignorance can be modeled by the uniform probability distribution.
Using this measure, what we want to specify as a security property in the example above is that the uncertainty of the coercer about the values of $V_A$ and $V_B$ should be maximal.
The set of properties of interest $\set{V_A, V_B}$ could have $2^2=4$ different combinations of values.
Therefore if we want that the coercer considers all of these combinations as possible, the Hartley measure of uncertainty of the coercer would be $\log 4 = 2$ bits.
To express this, we add a new operator $\uncertain$, and write the formula:
\begin{center}
  $\coop{v}\Sometm(Voted \wedge \uncertain^{\ge 2}_c\set{V_A,V_B})$
\end{center}

The formula states that the voter has a strategic ability to eventually cast her vote, while keeping the uncertainty of the coercer about the
valuations of $V_A$ and $V_B$ at the level of at least 2 bits.
Intuitively, the formula holds in state $s_1$ of model $M_2$, but not $M_1$.

In the next section, we use this idea to formalize the syntax and semantics of the logic $\ATLH$.

\section{\ATL with Uncertainty}\label{sec:athl}

In this section we define the syntax and semantics of the logic of strategic abilities with uncertainty operator $\ATLH$.
The logic is based on the idea of using the Hartley measure to quantify
the uncertainty of agents about the possible valuations of a set of formulas.
Similarly to $\ATLK$, the semantics of \ATLH is also defined with respect to concurrent epistemic game models (CEGM).

\subsection{Syntax}

The syntax of $\ATLH$ is given as follows:
\begin{center}
  $\varphi::= \prop{p} \mid \neg \varphi \mid \varphi\wedge\varphi
    \mid \coop{A}\Next\varphi \mid \coop{A}\Always\varphi \mid \coop{A}\varphi\Until\varphi \mid \uncertain_{a}^{\otimes\,m}\beta$.
\end{center}
where $A\in 2^{\Agt}$ is a set of players, $\beta\in 2^{\varphi}\setminus \set{\emptyset}$ is a set of formulas,
$a\in\Agt$ is a player, and $\otimes\in\set{<, \le, >, \ge, =}$ is a comparison operator.
For instance, the formula $\uncertain_{a}^{>m}\beta$ states that the the uncertainty of agent $a$ about the set of formulas $\beta$ is higher than $m$.

\subsection{Semantics}

Let $[q]_{\indist_a} = \set{q'\in\States\ \mid\ q'\indist_a q}$ denote the abstraction class of state $q\in\States$ with respect to relation $\indist_a$, i.e., the epistemic neighbourhood of $q$ from the perspective of agent $a\in\Agt$.
For a given formula $\varphi$, we define relation $\indist^{\varphi}\in \States\times\States$ that connects states with the same valuation of $\varphi$:
\begin{center}
  $q_1\sim^\varphi q_2\ \text{iff}\ M,q_1\models\varphi \Leftrightarrow M,q_2\models\varphi$.
\end{center}
If $\beta=\set{\varphi_1, ..., \varphi_n}$ is a set of formulas and $a\in\Agt$, then we define
\begin{center}
  $\sim^\beta_a = \sim_a \cap \bigcap^n_{i=1}\sim^{\varphi_i}$
\end{center}
I.e., $q_1 \sim^\beta_a q_2$ iff $q_1,q_2$ look the same to $a$ and cannot be discerned by any formula in $\beta$.
Note that $\sim^\beta_a$ is an equivalence.
We define
\begin{center}
  $\mathit{R}_{a,q}(\beta) = \set{[q']_{\sim_a^{\beta}}\ \mid\ q'\sim_a q}$
\end{center}
for the set of equivalence classes of $\sim_a^{\beta}$ contained in the epistemic neighbourhood of state $q$.
Then, the truth value of the statement ``agent $a$'s uncertainty about the set of formulas $\beta$ is in $\otimes$ relation to value $m\in\mathbb{R}$'' can be defined as follows:

\smallskip
\begin{center}
  $M,q\models \uncertain^{\otimes m}_a \beta,\  \text{iff}\ \log\vert\mathit{R}_{a,q}(\beta)\vert \otimes m$
\end{center}

\smallskip
Some straightforward validities of \ATLH are:
\begin{enumerate}
\item $\uncertain_a^{\substack{=m\\\geqslant m\\>m}}\beta \rightarrow \uncertain_a^{\substack{\geqslant m\\\geqslant m\\>m}}\beta'$,\quad for all $\beta\subseteq\beta'$;

\smallskip
\item $\uncertain_a^{\substack{<m\\\leqslant  m\\=m}}\beta \rightarrow \uncertain_a^{\substack{<m\\\leqslant  m\\\leqslant m}}\beta'$,\quad for all $\beta\subseteq\beta'$.
\end{enumerate}

\smallskip
Also, if $|\States|$ is the number of states in the model, then it holds that $M,q\models \uncertain_a^{<\min (|\beta|, \log (|\States|))}\beta$.

\subsection{Model Checking}\label{sec:mcheck}

In this section, we discuss model checking for \ATLH. The following results have long been known in the literature:
\begin{itemize}
    \item Model checking of epistemic logic is in \textbf{P} with respect to the size of the model and the length of the formula~\cite{halpern1991model}.
    \item Model checking of \ATLK for agents with \textit{ir} strategies is $\Deltwo$-complete with respect to size of the model and the length of the formula~\cite{Bulling10verification}.
    This is a direct consequence of the fact that model checking of $\textbf{ATL}_{ir}$ is \Deltwo-complete~\cite{Schobbens04ATL,Jamroga06atlir-eumas}.
\end{itemize}

In the following, we show that model checking of \ATLH is also \Deltwo-complete. To this end, it suffices to show that model checking of the uncertainty part of the language is in \textbf{P}.

\begin{proposition}\label{prop:mcheck-unc}
    If $\varphi$ is an \ATLH formula without strategic and temporal operators and $M$ is a CEGM,
    then checking if $\varphi$ is satisfied in a state $q$ of $M$ can be done in polynomial time with respect to $|\varphi|$ and $|M|$, where $|M|$ is the total number of states, transitions, and epistemic relation pairs in $M$.
\end{proposition}
\begin{proof}
    Let $\varphi_1, \varphi_2, ...\varphi_k$ be the subformulas of $\varphi$ (which incrementally generate $\varphi$) listed in order of length. We can see that $k\leq |\varphi|$, as there are at most $|\varphi|$ subformulas of $\varphi$.
    We start labeling each state in $M$ in increasing order of $i$, with labels $\varphi_i$ or $\lnot\varphi_i$, depending on whether $\varphi_i$ is true in that state or not.
    It is easy to see that we can do this in at most $\mathbf{O}(k |M|)$ labeling step.
    If the formula $\varphi_i$ is a propositional formula with respect to it's subformulas, then it can be labeled in in each state in constant time.
    In cases where $\varphi_i$ is of the form $\uncertain_a^{\otimes m}\beta$ where $\otimes\in\set{<, >, =}$ and $\beta = \set{\alpha_1, ...\alpha_{k'}}$, we have that each $\alpha_j$ is a subformula of $\varphi_i$.
    Therefore for labeling $\varphi_i$ we construct the set of equivalence classes $\mathit{R}_{a,q}(\beta)$ by checking the $k'$ labels of formulas in $\beta$ in all the states $q'$ where $q' \indist_a q$.
    Then we calculate $\log\vert\mathit{R}_{a,q}(\beta)\vert$ and compare it with $m$ in order to label $\varphi_i$.
    This procedure can be done in at most $\mathbf{O}(k' |M|)$ steps.
    Therefore the whole process of checking whether $\varphi$ is satisfied in a state $q$ or not can be done in at most $\mathbf{O}({\vert\varphi\vert}^2 {|M|}^2)$.
\end{proof}

\begin{proposition}
    Model checking of \ATLH for agents with \textit{ir} strategies is \Deltwo-complete with respect to size of the model and the length of the formula.
\end{proposition}
\begin{proof}
The lower bound follows from the fact that \ATLH subsumes $\textbf{ATL}_{ir}$, and model checking $\textbf{ATL}_{ir}$ is \Deltwo-hard.
The upper bound is straightforward from Proposition~\ref{prop:mcheck-unc} and the fact that model checking $\textbf{ATL}_{ir}$ is in \Deltwo.
\end{proof}

\section{Expressive Power of \ATLH}\label{sec:expressivity}

In this section we show that $\ATLH$ and $\ATEL$ have the same expressive power. We start by recalling the semantic definition of comparative expressivity~\cite{Wang09expressive}.

\begin{definition}[Expressivity]\label{def:expressivity}
  Let $L_1 = \langle \varPhi_1, \models_1, \mathbb{M} \rangle$ and $L_1 = \langle \varPhi_2, \models_2, \mathbb{M}\rangle$ represents two logics,
  such that $\varPhi_1$ and $\varPhi_2$ are the set of formulas defined in these logics, $\mathbb{M}$ is a nonempty class of models (or pointed models) over which the logics are defined,
  and $\models_1$ and $\models_2$ are the truth relations of these logics, such that $\models_1 \subseteq\mathbb{M}\times\varPhi_1$ and  $\models_2 \subseteq\mathbb{M}\times\varPhi_2$.
  We say that $L_2$ is at least as expressive as $L_1$ on the class of models $\mathbb{M}$, iff for every formula $\varphi_1\in\varPhi_1$, there exists a formula $\varphi_2\in\varPhi_2$
  such that for every $M\in\mathbb{M}$ we have $M\models_1 \varphi_1$ iff $M\models_2\varphi_2$. We will write it as $L_1\leqslant_\mathbb{M}L_2$.
\end{definition}

If both $L_1\leqslant_\mathbb{M}L_2$ and $L_2\leqslant_\mathbb{M}L_1$, then we say that $L_1$ and $L_2$ are equally expressive on $\mathbb{M}$, and write $L_1=_\mathbb{M}L_2$.

In the following, we use $\models_{\knows}$ and $\models_\uncertain$ to denote the semantic relation of $\ATEL$ and $\ATLH$, respectively, whenever it might not be clear from the context.

\subsection{Knowledge as Uncertainty}

\begin{theorem}
  $\ATLH$ is at least as expressive as $\ATEL$.
\end{theorem}
\begin{proof}
  Because the set of formulas defined in $\ATLH$ includes all the formulas defined in $\ATEL$ except the formulas
  including $\knows$ operator, and the semantics of the common formulas are similar in both logics, it suffices to prove that
  for any formula of type $\varphi_1 = \knows_a\varphi$ in $\ATEL$ there is a formula $\varphi_2$ in $\ATLH$ such that for every $M$,
  \begin{center}
    $M,q\models\knows_a\varphi\ \Leftrightarrow \ M,q\models_{\uncertain}\varphi_2$
  \end{center}

  We claim that we can construct such $\varphi_2$ from $\knows_a{\varphi}$ to be $\varphi_2 = \varphi \wedge \uncertain_a^{=0}\set{\varphi}$.
  Therefore we need to prove that:

  \begin{center}
    $M,q\models_{\knows}\knows_a\varphi\ \Leftrightarrow \ M,q\models_{\uncertain}\varphi \wedge \uncertain_a^{=0}\set{\varphi}$
  \end{center}

  We have that $M,q\models_{\knows}\knows_a\varphi$ if and only if $\varphi$ holds in all the indistinguishable states from $q$ for $a$, which includes state $q$ itself.
  This means that $\varphi$ holds in $q$ and $|R_{a,q}(\varphi)| = 1$, which in \ATLH would be expressed as $ M,q\models_{\uncertain}\varphi \wedge \uncertain_a^{=0}\set{\varphi}$.
\end{proof}

\subsection{Uncertainty as Knowledge}\label{sec:unc-as-knowledge}

\begin{theorem}\label{theorem:expessivity}
  \ATEL is at least as expressive as \ATLH.
\end{theorem}

The proof proceeds by translating every occurrence of $\uncertain_a^{\otimes m}\beta$ to a Boolean combination of epistemic formulas that express the knowledge of agent $a$ with respect to the \emph{indistinguishability classes} of the formulas in $\beta$, defined as follows:
\begin{definition}[Indistinguishability class of a formula]\label{def:indist1}
  For a given model $M$, if $q\in\States$, $a\in\Agt$ and $\varphi$ is a formula,
  then we define the indistinguishability class of $\varphi$ with respect to $q$ and $a$ as follows:
  \begin{center}
    $[\varphi]_a^q = [\varphi] \cap [q]_{\indist_a}$,
  \end{center}
  where $[q]_{\indist_a}$ denotes the set of states that are indistinguishable from $q$ for $a$,
  and $[\varphi]$ is the set of states $q'\in\States$ were $M,q'\models \varphi$.
\end{definition}

The full proof is technical and rather tedious; it can be found in the \short{extended version of the paper~\cite{Tabatabaei23uncertainty-arxiv}}\extended{appendix}.
Here, we present how the translation works on an example.
Let $\varphi_1$ and $\varphi_2$ be two formulas that do not contain any $\uncertain$ operators.
  We would like to find an \ATEL formula $P(\varphi_1, \varphi_2)$, such that:
  \begin{center}
    $M,q\models \uncertain_a^{=\log 3}\set{\varphi_1, \varphi_2}\ \Leftrightarrow\ M,q\models_{\knows}P(\varphi_1,\varphi_2)$
  \end{center}

  First we define new formulas $A$, $B$, $C$ and $D$ as follows:
  $A= \varphi_1 \wedge \varphi_2$,
  $B= \varphi_1 \wedge \lnot\varphi_2$,
  $C= \lnot\varphi_1 \wedge \varphi_2$ and
  $D= \lnot\varphi_1 \wedge \lnot\varphi_2$.
  It is clear that the sets of states $[A]_a^q$, $[B]_a^q$, $[C]_a^q$ and $[D]_a^q$ are mutually exclusive, and moreover they partition $[q]_{\indist_a}$.
  Because the truth value of each one of $A$, $B$, $C$, $D$ corresponds to the truth value of exactly one of four possible different valuation combinations of $\varphi_1$ and $\varphi_2$, so they are distinct.

  If $M,q\models \uncertain_a^{=\log 3}\set{\varphi_1, \varphi_2}$, then exactly one of $[A]_a^q$, $[B]_a^q$, $[C]_a^q$ or $[D]_a^q$ has to be empty.
  Because these sets are mutually disjoint, if all are non-empty then we should have at least four different states in $[q]_{\indist_a}$ with the four different valuation combinations for
  the formulas $\varphi_1$ and $\varphi_2$. This would mean that $M,q\models \uncertain_a^{=\log 4}\set{\varphi_1, \varphi_2}$
  which contradicts $M,q\models \uncertain_a^{=\log 3}\set{\varphi_1, \varphi_2}$. Similarly if more than one of  $[A]_a^q$, $[B]_a^q$, $[C]_a^q$ or $[D]_a^q$ are empty,
  then it means that only two or less possible valuation combinations of $\varphi_1$ and $\varphi_2$ exist in $[q]_{\indist_a}$.
  This entails that $M,q\models \uncertain_a^{<\log 3}\set{\varphi_1, \varphi_2}$, which is again a contradiction.
  The converse is also true: if exactly three of the sets $[A]_a^q$, $[B]_a^q$, $[C]_a^q$ or $[D]_a^q$ are non-empty, then there are exactly three
  valuation combinations of $\varphi_1$ and $\varphi_2$ in $[q]_{\indist_a}$, which follows that $M,q\models \uncertain_a^{=\log 3}\set{\varphi_1, \varphi_2}$.
  So the formula $M,q\models \uncertain_a^{=\log 3}\set{\varphi_1, \varphi_2}$ holds if and only if exactly one of $[A]_a^q$, $[B]_a^q$, $[C]_a^q$ or $[D]_a^q$ is empty.
  This can happen in four different ways (one corresponding to each of  $[A]_a^q$, $[B]_a^q$, $[C]_a^q$ or $[D]_a^q$ being empty).

  First consider the case where $[C]_a^q$ is empty. Then:
  \begin{align*}
                    & \nexists q' \textit{ s.t }\ q'\indist_a q \textit{ and } M,q'\models C \\
    \Leftrightarrow & (\forall q',\   q'\indist_a q \ \Longrightarrow  M,q'\not\models C)    \\
    \Leftrightarrow & (\forall q',\   q'\indist_a q \ \Longrightarrow  M,q'\models \lnot C)  \\
    \Leftrightarrow & M,q\models \knows_a \lnot C
  \end{align*}

  In a similar way we can show that $[A]_a^q$ is non-empty iff $M,q\models \lnot \knows_a \lnot A$. The same goes for  $[B]_a^q$, and $[D]_a^q$.
  Therefore among  $[A]_a^q$, $[B]_a^q$, $[C]_a^q$ and $[D]_a^q$, only $[C]_a^q$ is empty iff:

  \begin{center}
    $M,q\models \lnot \knows_a \lnot A \wedge \lnot \knows_a \lnot B \wedge \knows_a \lnot C \wedge \lnot \knows_a \lnot D$
  \end{center}

  We got this result by assuming that only $[C]_a^q$ is empty. Given that $M,q\models \uncertain_a^{=\log 3} \set{\varphi_1, \varphi_2}$
  iff exactly one of $[A]_a^q$, $[B]_a^q$, $[C]_a^q$ or $[D]_a^q$ is empty, and knowing that we have four possible choices for which one is to be empty, we get that:
  \begin{center}
    $M,q\models \uncertain_a^{=\log 3}\set{\varphi_1, \varphi_2} \Leftrightarrow M,q\models_{\knows}P(\varphi_1, \varphi_2)$
  \end{center}

  where $P(\varphi_1, \varphi_2)$ is defined as:
  \begin{align*}
  P(\varphi_1, \varphi_2) =\quad & ( \knows_a \lnot (\varphi_1 \wedge \varphi_2) \wedge \lnot \knows_a \lnot (\varphi_1 \wedge \lnot\varphi_2)                              \\
         & \qquad \wedge \lnot \knows_a \lnot (\lnot\varphi_1 \wedge \varphi_2) \wedge \lnot \knows_a \lnot (\lnot\varphi_1 \wedge \lnot\varphi_2)) \\
    \vee & ( \lnot \knows_a \lnot (\varphi_1 \wedge \varphi_2) \wedge \knows_a \lnot (\varphi_1 \wedge \lnot\varphi_2)                              \\
         & \qquad\wedge \lnot \knows_a \lnot (\lnot\varphi_1 \wedge \varphi_2) \wedge \lnot \knows_a \lnot (\lnot\varphi_1 \wedge \lnot\varphi_2))  \\
    \vee & ( \lnot \knows_a \lnot (\varphi_1 \wedge \varphi_2) \wedge \lnot \knows_a \lnot (\varphi_1 \wedge \lnot\varphi_2)                        \\
         & \qquad\wedge \knows_a \lnot (\lnot\varphi_1 \wedge \varphi_2) \wedge \lnot \knows_a \lnot (\lnot\varphi_1 \wedge \lnot\varphi_2))        \\
    \vee & ( \lnot \knows_a \lnot (\varphi_1 \wedge \varphi_2) \wedge \lnot \knows_a \lnot (\varphi_1 \wedge \lnot\varphi_2)                        \\
         & \qquad\wedge \lnot \knows_a \lnot (\lnot\varphi_1 \wedge \varphi_2) \wedge \knows_a \lnot (\lnot\varphi_1 \wedge \lnot\varphi_2)).
  \end{align*}

\section{Uncertainty Is Exponentially More Succinct than Knowledge}\label{sec:succinctness}

The notion of succinctness~\cite{Stockmeyer72phd,Wilke99succinctness,Adler01succinctness} is a refinement of the notion of expressivity. Assume that one particular property can be expressed in both languages $L_1$ and $L_2$, with
formulas $\varphi_1$ and $\varphi_2$ respectively. When comparing the representational succinctness of these two languages, we are interested in whether there is a significant
difference in the lengths of $\varphi_1$ and $\varphi_2$. Similar to analysis of complexity, what we consider \textit{significant} is at least exponential growth of the size of
a formula in one of the languages comparing to the equivalent formula in the other language. In this section, we prove that the language of $\ATLH$ is
exponentially more succinct than $\ATEL$. We use the so called \emph{formula size games (FSG)} from~\cite{French13succinctness} to construct the proof.
In brief, we will show that for any $n\in\mathbb{N}$, there is a formula $\varphi_n$ of size linear to $n$ in \ATHL, such that for any formula $\varphi_n'$ in
$\ATEL$ with the same set of satisfying models as $\varphi_n$, the parse tree of $\varphi_n'$ will have at least $2^n$ distinct nodes, and therefore the size of $\varphi_n'$ is at least exponential in $n$.

\subsection{Succinctness and Formula Size Games}\label{sec:formulasizegame}

Before showing that \ATL with uncertainty is exponentially more succinct than \ATL with knowledge, we summarize the basic terminology.

\begin{definition}[Length of formulas in \ATLH]\label{def:formulasize}
  The length of formula $\varphi$ is denoted by $|\varphi|$, recursively defined as follows:\\
  $|p| = 1$,\\
  $|(\varphi_1 \vee \varphi_2)| = |\varphi_1 \Until \varphi_2| = |\varphi_1| + |\varphi_2| + 1$,\\
  $|\lnot\varphi| = |\Next\varphi| = |\Always\varphi| = 1 + |\varphi|$\\
  $|\coop{A}\varphi| = |A| + |\varphi|$,\\
  $|\uncertain_a^{\otimes m}\beta| = 1 + \sum_{\varphi_i\in\beta} |\varphi_i|$.
\end{definition}

\begin{definition}[Succinctness]\label{def:succinctness}
  Let $L_1=\langle \varPhi_1, \models_1, M \rangle$ and  $L_2=\langle \varPhi_2, \models_2, M \rangle$ be two logics such that  $L_1\leqslant_M L_2$
  and let $f(x) = O(g(n))$ be a strictly increasing function.  If for every $n\in \mathbb(N)$ there are two formulas $\alpha_n\in \varPhi_1$
  and $\beta_n \in \varPhi_2$ where:
  \begin{itemize}
    \item $|\alpha_n| = f(n)$
    \item $|\beta_n| = 2^{g(n)}$
    \item $\beta_n$ is the shortest formula on $\varPhi_2$ that is equivalent to $\alpha_n$ on $M$,
  \end{itemize}
  then we say that $L_1$ is exponentially more succinct than $L_2$ on $M$ and write it as: $L_1 \not\leqslant_M^\mathit{subexp} L_2$.
\end{definition}

In the following, for a set of pointed models $A$, we use the term $A\models \varphi$ to mean that $\forall M\in A \,.\, M\models\varphi$.

\begin{definition}[FSG]\label{def:FSG}
One-person formula size game (FSG) on two sets of pointed models $A$ and $B$ is played as follows: during the course of the game, a game tree is constructed
such that each node is labeled with pair $\pairModels{C}{D}$ of sets of pointed models.
The possible moves for the player (called \emph{the spoiler}) on each node of the tree are $\set{p\in\atomicP, \lnot, \vee, \knows_i}$, where $i\in\Agt$. A node can be open or closed.
Once a node is closed, no further move can be played there. The condition and consequences of each of possible moves are as below:
\begin{description}
\item[$p\in\atomicP$] (\textit{Atomic} move): the spoiler chooses $p\in\atomicP$ such that $C\models p$ and $D\models \lnot p$. Then the node is declared closed.

\item[$\lnot$] (\textit{Not} move): A new node $\pairModels{C}{D}$ is added to the tree.

\item[$\vee$] (\textit{Or} move): two nodes $\pairModels{C_1}{D}$ and $\pairModels{C_2}{D}$ are added to the tree such that $C_1 \cup C_2 = C$.

\item[$K_i$] (\textit{Knows} move), where $i\in\Agt$: For each $(M, s)\in D$ the spoiler chooses a pointed model $(M, s')$ such that $s\indist_a s'$.
If for some $(M, s)\in D$ such $(M, s')$ does not exist, then the spoiler cannot play this move. All such chosen pointed models are collected
in $D'$. Moreover, for each $(M, s)\in C$, all possible pointed models $(M,s')$ such that $s\indist_a s'$ are added to $C'$. Then a new node
$\pairModels{C'}{D'}$ is added to the tree.
\end{description}
We say that the spoiler wins FSG starting at $\pairModels{A}{B}$ in $n$ moves iff there is a game tree $T$ with root $\pairModels{A}{B}$
and precisely n nodes such that every leaf of $T$ is closed.
\end{definition}

\begin{theorem}[\cite{French13succinctness}]\label{theorem:formulaSizeGame}
  The spoiler can win the FSG starting at $\pairModels{A}{B}$ in less than $k$ moves iff there is some $n<k$ and a formula $\varphi\in\varPhi_{MEL}$
  such that $A\models_{MEL} \varphi$, $B\models_{MEL} \lnot \varphi$ and $|\varphi|=n$, where $\varPhi_{MEL}$ is the set of formulas defined in \textit{Multiagent epistemic logic}
  and $\models_{MEL}$ shows truth relation in it.
\end{theorem}

The game tree through which the spoiler wins the FSG is the parse three of formula $\varphi$ in the language of \textit{Multiagent epistemic logic}.
For any $\pairModels{A}{B}$, the set of all closed game trees with root $\pairModels{A}{B}$ is denoted by $T(\pairModels{A}{B})$.
Consequently, the set of closed trees represents also the set of all formulas $\varphi$ that could distinguish the set of pointed models $A$
from the set of pointed models $B$ via the truth relation $\models_{MEL}$.

\subsection{\ATLH Is More Succinct than \ATLK}\label{sec:succinctness-detailed}

Theorem~\ref{theorem:formulaSizeGame} allows us to use FSG for proving the succinctness of our new logic \ATLH with respect to $\ATEL$.

\begin{theorem}\label{theorem:succinctness}
  The logic $\ATHL$ is exponentially more succinct than the logic $\ATEL$.
\end{theorem}

The full proof is rather technical; it can be found in the \short{extended version of the paper~\cite{Tabatabaei23uncertainty-arxiv}}\extended{appendix}. Here we explain the sketch of the proof:

\begin{proof}[Proof sketch]
Let $Lang(\ATEL)$ and $Lang(\ATHL)$ represent the languages of the logics $\ATEL$ and $\ATHL$ respectively.
For every $n\in \mathbb{N}$, we define a formula $\varphi_n\in Lang(\ATHL)$ where $f(n) = n$.
Then, we define two sets of pointed models $A_n$ and $B_n$, such that $A_n\models_{\uncertain}\varphi_n$ and $B_n\models_{\uncertain}\lnot\varphi_n$.
Because $\ATEL$ is as expressive as $\ATHL$, there exists a formula $\psi_n\in Lang(\ATEL)$ which is the shortest formula in $Lang(\ATEL)$ equivalent to $\varphi_n$.
Therefore $\psi_n$ too can distinguish sets $A_n$ and $B_n$,
which means the spoiler can win the FSG game starting at $\pairModels{A_n}{B_n}$ by playing the formula $\psi_n$.
We then prove that the spoiler cannot win the FSG game starting at $\pairModels{A_n}{B_n}$ in less than $2^n$ moves, which means the size of $\psi_n$ is at least $2^n$.
\end{proof}

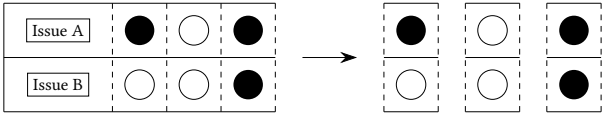
\begin{figure}[t]\centering
\resizebox{0.45\textwidth}{!}{\begin{tikzpicture}[node distance = 4cm]

\node[draw, align = right] at (-1, 1.5) {Issue A};
      \node[draw, align = right] at (-1, 0.5) {Issue B};

      \draw
      (-2,0) -- (3,0)
      (-2,1) -- (3,1)
      (-2,2) -- (3,2);

      \draw (-2,0) -- (-2,2);

      \draw[dashed]
      (0,0) -- (0,2)
      (1,0) -- (1,2)
      (2,0) -- (2,2)
      (3,0) -- (3,2);

      \draw (0.5,0.5) circle (2.0ex);
      \fill (0.5,1.5) circle (2.0ex);
      \draw (1.5,0.5) circle (2.0ex);
      \draw (1.5,1.5) circle (2.0ex);
      \fill (2.5,0.5) circle (2.0ex);
      \fill (2.5,1.5) circle (2.0ex);

      \draw [-{Stealth[length=3mm, width=2mm]}] (3.5,1) -- (4.5,1);

\draw
      (5,0) -- (6,0) (6.5,0) -- (7.5,0) (8,0) -- (9,0)
      (5,1) -- (6,1) (6.5,1) -- (7.5,1) (8,1) -- (9,1)
      (5,2) -- (6,2) (6.5,2) -- (7.5,2) (8,2) -- (9,2);

      \draw[dashed]
      (5,0) -- (5,2) (6,0) -- (6,2)
      (6.5,0) -- (6.5,2) (7.5,0) -- (7.5,2)
      (8,0) -- (8,2) (9,0) -- (9,2);

      \draw (5.5,0.5) circle (2.0ex);
      \fill (5.5,1.5) circle (2.0ex);
      \draw (7,0.5) circle (2.0ex);
      \draw (7,1.5) circle (2.0ex);
      \fill (8.5,0.5) circle (2.0ex);
      \fill (8.5,1.5) circle (2.0ex);

    \end{tikzpicture}
  }
  \caption{Example ThreeBallot vote. Here the voter has voted for issue $A$ and against issue $B$. Thus, the value of the vote is $A\notB$. The voter fills two fields in the row of
  issue $A$ and one space in the row of issue $B$, and then she separates the three ballots. The resulting ballot set is $\set{FB, BB, FF}$.}
  \label{fig:threeballot}
\end{figure}

\begin{table}[t]\scalebox{0.65} {
    \mbox{}\!\!\!\!\!\!\!\!
    \begin{tabular}{l|c|r}
        \textbf{Vote and ballot set (BS)} & \textbf{Receipt} & \textbf{Possible information sets of the coercer}\\
        \hline
        \hline
        \multirow{3}{*}{Vote = $\notA\notB$, BS = $\set{BB, FB, BF}$}
        & BB & $\set{\notA\notB}$, $\set{A \notB, \notA\notB}$, $\set{\notA B, \notA\notB}$, $\set{A \notB, \notA B, \notA\notB}$, $\set{A \notB, \notA B, AB, \notA\notB}$ \\
        & FB & $\set{\notA\notB}$, $\set{\notA B, \notA\notB}$,  $\set{A \notB, \notA\notB}$, $\set{A \notB, AB, \notA\notB}$, $\set{A \notB, \notA B, AB, \notA\notB}$\\
        & BF & $\set{\notA\notB}$, $\set{\notA B, \notA\notB}$, $\set{A \notB, \notA\notB}$, $\set{\notA B, AB, \notA\notB}$, $\set{A \notB, \notA B, AB, \notA\notB}$\\
        \hline
        \multirow{3}{*}{Vote = $\notA\notB$, BS = $\set{BB, BB, FF}$}
        & BB & $\set{\notA\notB}$, $\set{\notA B, \notA\notB}$, $\set{A \notB, \notA\notB}$, $\set{AB, \notA\notB}$, $\set{A \notB, \notA B, AB, \notA\notB}$\\
        & FF & $\set{\notA\notB}$, $\set{\notA B, \notA\notB}$, $\set{A \notB, \notA\notB}$, $\set{AB, \notA\notB}$, $\set{A \notB, \notA B, AB, \notA\notB}$\\
         \hline
         \hline
        \multirow{3}{*}{Vote = $A \notB$, BS = $\set{BB, FB, FF}$}
        & BB & $\set{A \notB}$, $\set{A \notB, AB}$, $\set{A \notB, \notA\notB}$, $\set{A \notB, \notA B, \notA\notB}$, $\set{A \notB, \notA B, AB, \notA\notB}$\\
        & FB & $\set{A \notB}$, $\set{A \notB, AB}$, $\set{A \notB, \notA\notB}$, $\set{A \notB, AB, \notA\notB}$, $\set{A \notB, \notA B, AB, \notA\notB}$\\
        & FF & $\set{A \notB}$, $\set{A \notB, AB}$, $\set{A \notB, \notA\notB}$, $\set{A \notB, \notA B, AB}$, $\set{A \notB, \notA B, AB, \notA\notB}$\\
        \hline
        \multirow{3}{*}{Vote = $A \notB$, BS = $\set{FB, FB, BF}$}
        & FB & $\set{A \notB}$, $\set{A \notB, \notA B}$, $\set{A \notB, AB}$, $\set{A \notB, \notA\notB}$, $\set{A \notB, \notA B, AB, \notA\notB}$\\
        & BF & $\set{A \notB}$, $\set{A \notB, \notA B}$, $\set{A \notB, AB}$, $\set{A \notB, \notA\notB}$, $\set{A \notB, \notA B, AB, \notA\notB}$\\
        \hline
        \hline
        \multirow{3}{*}{Vote = $\notA B$, BS = $\set{BB, BF, FF}$}
        & BB & $\set{\notA B}$, $\set{\notA B, \notA\notB}$, $\set{\notA B, AB}$, $\set{A \notB, \notA B, \notA\notB}$, $\set{A \notB, \notA B, AB, \notA\notB}$\\
        & BF & $\set{\notA B}$, $\set{\notA B, AB}$, $\set{\notA B, \notA\notB}$, $\set{\notA B, AB, \notA\notB}$, $\set{A \notB, \notA B, AB, \notA\notB}$\\
        & FF & $\set{\notA B}$, $\set{\notA B, \notA\notB}$, $\set{\notA B, AB}$, $\set{A \notB, \notA B, AB}$, $\set{A \notB, \notA B, AB, \notA\notB}$\\
        \hline
        \multirow{3}{*}{Vote = $\notA B$, BS = $\set{FB, BF, BF}$}
        & FB & $\set{\notA B}$, $\set{\notA B, \notA\notB}$, $\set{\notA B, AB}$, $\set{A \notB, \notA B}$, $\set{A \notB, \notA B, AB, \notA\notB}$\\
        & BF & $\set{\notA B}$, $\set{\notA B, \notA\notB}$, $\set{\notA B, AB}$, $\set{A \notB, \notA B}$, $\set{A \notB, \notA B, AB, \notA\notB}$\\
        \hline
        \hline
        \multirow{3}{*}{Vote = $AB$, BS = $\set{FB, BF, FF}$}
        & FB & $\set{AB}$, $\set{A \notB, AB}$, $\set{\notA B, AB}$, $\set{A \notB, AB, \notA\notB}$, $\set{A \notB, \notA B, AB, \notA\notB}$\\
        & BF & $\set{AB}$, $\set{A \notB, AB}$, $\set{\notA B, AB}$, $\set{\notA B, AB, \notA\notB}$, $\set{A \notB, \notA B, AB, \notA\notB}$\\
        & FF & $\set{AB}$, $\set{A \notB, AB}$, $\set{\notA B, AB}$, $\set{A \notB, \notA B, AB}$, $\set{A \notB, \notA B, AB, \notA\notB}$\\
        \hline
        \multirow{3}{*}{Vote = $AB$, BS = $\set{BB, FF, FF}$}
        & BB & $\set{AB}$, $\set{AB, \notA\notB}$, $\set{A \notB, AB}$, $\set{\notA B, AB}$, $\set{A \notB, \notA B, AB, \notA\notB}$\\
        & FF & $\set{AB}$, $\set{AB, \notA\notB}$, $\set{A \notB, AB}$, $\set{\notA B, AB}$, $\set{A \notB, \notA B, AB, \notA\notB}$\\
        \hline
        \hline
      \end{tabular}
    }
  \bigskip
  \caption{Indistinguishability sets of the coercer in a two voter, two issue referendum with ThreeBallot, based on the selected receipt by the voter.
    Each row lists the possible epistemic classes $[q]_{\sim_c}$ of the coercer, depending on the voter's vote and the way she filled her ballots (invisible to the coercer), the voter's election receipt (visible to the coercer), and the set of ballots posted on the public bulletin board }
\label{table:threeballot}
  \end{table}

\section{Case Study: Three Ballot}\label{sec:threeballot}

We demonstrate the usefulness of our proposal on a real-life voting scenario.

\subsection{Voting with ThreeBallot}

ThreeBallot~\cite{Rivest06threeBallot,Rivest07threeProtocols} has been proposed by Rivest as a paper-based end-to-end verifiable voting protocol.
Here, we use a simplified version of the protocol, which can be used for multiple-issue referendums.
Following the example in Section~\ref{sec:motivating}, we consider a two-issue referendum, in which each voter votes to accept or reject two proposals $A$ and $B$.

In ThreeBallot, the ballots are prepared such that for each issue, there are three empty fields that the voter can fill.
For voting to accept the issue, the voter has to fill exactly two of the empty fields, and to vote to reject the issue, the voter has to fill exactly one empty field.
However, the exact positions of the filled spaces are up to the voter. After filling the ballots, the voter separates three columns of fields, and in this way creates three separate ballots.
The voter can make (and keep) a copy of one of those ballots as the receipt. Finally, she puts all the original ballots in the ballot box.
The tally of the votes is done by counting the filled spaces for each issue.
The difference between the number of filled spaces for each issue and the number of voters shows the number of votes in favor of that issue.
After the tallying, all the ballots are published on a public bulletin board, so that everyone can check the correctness of the result.

One of the main goals of ThreeBallot is \emph{coercion resistance}~\cite{Juels05coercion}, i.e., the protocol should make it impossible for a third party to successfully coerce or bribe voters into voting in a particular way.
Informally, coercion resistance is usually understood as the inability of the coercer to learn how the voter has voted, even if the voter cooperates with the coercer.
Interestingly, ThreeBallot was both claimed secure~\cite{Rivest07threeProtocols} and insecure~\cite{Appel06threeBallot}.
It might seem that one of the claims must be wrong, but a closer look reveals that they are based on different concepts of vote privacy.
In~\cite{Rivest07threeProtocols}, it is argued that the coercer cannot get to \emph{know how the voter has voted}, which is a strategic-epistemic property.
In contrast, \cite{Appel06threeBallot}~argues that the coercer can \emph{gain information} about the value of the vote.
We demonstrate the difference in the remainder of this section.

\subsection{Model of the Scenario}

In our example, the two propositions that determine the vote of the voter are $V_A$ and $V_B$.
We encode a vote against an issue by using \textit{overline}. So, $\notA B$ indicates a vote for issue $B$ and against issue $A$; in other words, it states that $V_A$ is \textit{false} and $V_B$ is \textit{true}.
This way, the set of possible votes from the voter is $Votes = \set{AB, \notA B, A\notB, \notA\notB}$.
Similarly, we encode a filled space by $F$ and a blank (not filled) space by $B$. For instance, a $\set{Blank, Filled}$ ballot is denoted by $BF$.
Figure~\ref{fig:threeballot} depicts an example ThreeBallot card and the resulting ballots.

We consider a scenario with two voters and a two issue referendum with issues $A$ and $B$.
Let us call the first voter \textit{the voter}, as she will be the one that we are focusing on in this example. We call the second voter \textit{the other voter}.
During the election, the voter has several choices. The first is what vote she is going to cast.
Then for each possible vote, there are various ways that the ballots can be filled, which will result in different ballot sets.
After that, the voter has the choice of which of the ballots to keep a copy as the receipt.
In our scenario, there exist a coercer who after the election will force the voter to reveal her receipt.
The coercer then tries to infer the actual value of the voter's vote, based on the receipt and the published bulletin board.

Table~\ref{table:threeballot} shows the different ways how the choices of the voter affect the possible indistinguishability set of the coercer about the value of the vote, after the receipt has been revealed.
The different indistinguishability sets in each row result from various ways in which the other voter might fill his ballot.

\subsection{Analysis: Epistemic Security}\label{sec:epist-security}

The coercion resistance security property is usually framed following the idea that the coercer cannot \emph{get to know} how the voter has voted, even if the voter cooperates with the coercer.
In ~\cite{Tabatabaei16expressing}, various nuances of coercion resistance are formulated in the logic \ATLK.
In a similar way, we can use \ATLK to express coercion resistance in our ThreeBallot example as follows:
\begin{center}
    $\bigwedge_{V_i\in Votes} \lnot\coop{v,c}\Sometm(V_1 = V_i \wedge (V_1 \neq V_2 \rightarrow \knows_c (V_1 = V_i)))$
\end{center}
The formula states that for any vote choice, there exists no common strategy for the voter and the coercer, such that the voter selects that vote and given that the choice of the two voters are different (the reason for this condition is explained below), the coercer would know the value of the vote.\footnote{
  Note that we use $(V_1 = V_i)$ in the role of an atomic proposition which evaluates to true whenever $V_1$ is indeed equal to $V_i$. Condition $(V_1 \neq V_2)$ is treated analogously. }

It is obvious that in the cases where both voters have voted identically, even without revealing the receipt, the coercer will know the value of the vote just by looking at the bulletin board.
 This is similar to the case when in an election all the voters vote similarly, in which case the privacy of their votes will be broken after publishing the tally (unless some sort of obfuscation is used~\cite{Jamroga19tallies}).
 We added the condition $(V_1 \neq V_2)$ to the above formula to account for this. Also in the following we only focus on the cases where the two voters has voted differently.

 By looking at Table~\ref{table:threeballot} we can see that the model satisfies the coercion resistance property as formulated in the above \ATLK expression.
 This is because there is no row in the table that consists of only one indistinguishability set for the coercer which has only one member (the actual value of the vote).
 However the voter votes and selects the receipt, there is at least one possible indistinguishability set with more than one member, meaning that the coercer might not get to know the actual vote of the voter.

\subsection{Information-Theoretic Security in \ATLH}\label{sec:info-security}

We can alternatively define the coercion resistance property in the information-theoretic sense, namely that the coercer cannot \emph{gain information} on how the voter has voted, even if the voter cooperates with the coercer.
Phrasing this differently, we want that no matter the course of actions of the voter and coercer, the coercer has always maximum uncertainty about the actual value of the vote.
We can express this property in \ATLH as follows:

\begin{center}
    $\bigwedge_{V_i\in Votes} \lnot\coop{v,c}\Sometm(V_1 = V_i \wedge (V_1 \neq V_2 \rightarrow \uncertain^{=\log(\vert Votes \vert)}_c \set{V_A, V_B}))$
\end{center}

The above formula states that, for any joint strategy of the coercer and the voter, the uncertainty of the coercer will always be at the maximum.
Looking at Table~\ref{table:threeballot}, we can see that the ThreeBallot protocol does not satisfy this property. This is because in each row there exists a possible indistinguishability set whose size is less than the number of possible votes.

This example shows that, although ThreeBallot could be considered secure with respect to the epistemic notion of coercion resistance expressed in \ATLK,
it is not secure when we define the security requirement as an information-theoretic property, and formalize it in \ATLH.

\section{Conclusion}\label{sec:conclusion}

In this work, we introduce the logic \ATLH which extends alternating-time temporal logic \ATL with quantitative modalities based on the Hartley measure of uncertainty.
As the main technical result, we show that \ATLH has the same expressive power and the same model checking complexity as \ATLK (i.e., \ATL with epistemic modalities), but it is exponentially more succinct.

The succinctness result, together with the model checking complexity, is of major significance.
As we have seen in Section~\ref{sec:mcheck}, both \ATLK and \ATLH have the same verification complexity with respect to the size of the model and \textit{the length of the formula}.
Theorem~\ref{theorem:succinctness} promises that, for some properties, their verification in \ATLH will be exponentially faster than in \ATLK.
Also, a more succinct language often results in better readability, which in turn helps the designers of a system to make less mistakes in the specification of desired properties.
Last but not least, many properties can be expressed in \ATLH in a much more intuitive way than in \ATLK. Understanding the information-theoretic intuition of a corresponding \ATLK formula can be a real challenge.

We suggest the specification of security requirements as an important application area for our proposal. In particular, the framework can be used to expose the logical structure of security claims, for example, the difference between the epistemic and information-theoretic notions of privacy. We demonstrate this on a real-life voting scenario involving the ThreeBallot protocol, which has been both claimed secure and insecure in the past. Indeed, the protocol is secure with respect to an epistemic notion of privacy, but it may fail to obtain the information-theoretic one.

In the future, we plan to implement model checking for \ATLH as an extension of the STV~\cite{Kurpiewski21stv-demo} or MCMAS~\cite{Lomuscio17mcmas} model checkers.

\begin{acks}
The work has been supported by NCBR Poland and FNR Luxembourg under the PolLux/FNR-CORE projects STV (POLLUX-VII/1/2019 \& C18/IS/12685695/IS/STV/Ryan) and SpaceVote (POLLUX-XI/14/\-SpaceVote/2023).
\end{acks}

\bibliographystyle{ACM-Reference-Format}

\extended{
  \appendix
  \section{Detailed Proofs}

\subsection*{Proof of Theorem~\ref{theorem:expessivity}}
\begin{theorem*}
    \ATEL is at least as expressive as \ATLH.
\end{theorem*}

Because the set of formulas defined in $\ATEL$ includes all the formulas defined in $\ATLH$, except the formulas
including $\uncertain$ operator, and the semantics of the common formulas are similar in both logics,
we formulate the proof by giving a translation from an \ATHL formula of type $\uncertain_a^{\otimes m}\beta$ to an equivalent \ATEL formula.
Before giving the translation, we introduce some helper definitions:
For a given formula $\varphi$, let ${(\lnot)}^0:=\lnot \varphi$ and  ${(\lnot)}^1:=\varphi$.

\begin{definition}\label{def:partitioningSet}
  If $\beta = \set{\varphi_1, ...\varphi_n}$ is a set of $n$ formulas, we define the set of formulas $\varPhi_\beta$ as
  \begin{center}
    $\varPhi_\beta = \set{(\lnot)^{t_1}\varphi_1 \wedge (\lnot)^{t_2}\varphi_2 \wedge \  ... \  \wedge (\lnot)^{t_n} \varphi_n \, \mid \, t_i\in \set{0,1}}$
  \end{center}
\end{definition}

Each member of $\varPhi_\beta$ is formula which is constructed by the conjunction of $n$ terms, each term being either a member of $\beta$ or its negation.
Because $\beta$ has $n$ members, the size of $\varPhi_\beta$ is $2^n$.
It can be seen that for any $q\in\States$ and $a\in\Agt$, the indistinguishablity classes of members of $\varPhi_\beta$ with respect to $q$ and $a$ (Definition.\ref{def:indist1})
are mutually exclusive and they partition $[q]_{\indist_a}$ (with the same reasoning as explained in Section~\ref{sec:unc-as-knowledge}).

\begin{definition}
  For $m,n\in\mathbb{N}$, and $m \leqslant 2^n$, we define the set $T_{n,m}$ as follows:
  \begin{center}
    $T_{n,m} = \set{(s_1,s_2,...,s_{2^n})\mid \forall i . s_i\in\set{0,1} ,\, \sum_{i = 1}^{2^n}s_i = 2^n - m}$
  \end{center}
\end{definition}
$T_{n,m}$ is the set of $\binom{2^n}{m}$ possible orderings of $0$s and $1$s, such that the number of $1$s is $2^n-m$ and the number of $0$s is $m$.
If $t\in T_{n,m}$, then we refer to $j$th element of $t$ by $t(j)$ (which of course could be either $0$ or $1$).

\begin{proposition}\label{prop:expressivity}
  If $\beta = \set{\varphi_1, ...\varphi_n}$ is a set of $n$ formulas
  then for a model $M$, $q\in \States$ and $a\in\Agt$ we have
  \begin{center}
    $M,q\models \uncertain_a^{=\log m}\beta \ \Longleftrightarrow \ M,q\models_{\knows}P_m(\varphi_1,...,\varphi_n)$
  \end{center}

  where $1\leq m \leq 2^n$ and $P_m(\varphi_1, ..., \varphi_n)$ is defined as below:
  \begin{center}
    $P_m(\varphi_1, ..., \varphi_n) = \bigvee_{t_i\in T_{n,m}}((\lnot)^{t_i(1)}\knows_a\lnot \alpha_1 \wedge (\lnot)^{t_i(2)}\knows_a\lnot \alpha_2 \wedge ... \wedge (\lnot)^{t_i(2^n)}\knows_a\lnot \alpha_{2^n})$,
  \end{center}
  and $(\alpha_1, \alpha_2,..., \alpha_{2^n})$ is any ordering on the members of $\varPhi_\beta$.

  \begin{proof}
    The idea of the proof is similar to the example in Section~\ref{sec:unc-as-knowledge}.
    By the way we constructed $\varPhi_\beta$, we have that the members of the set $\omega =\set{[\alpha_i]_a^q\mid 1\leqslant i\leqslant 2^n}$ are mutually exclusive,
    and they partition $[q]_{\indist_a}$. Therefore $M,q\models \uncertain_a^{=\log m}\beta$ if and only if exactly $m$ members of the set $\omega$ are non-empty.
    For any $i$, we have that $[\alpha_i]_a^q$ is empty if and only if $M,q\models \knows_a \lnot\alpha_i$, and similarly $[\alpha_i]_a^q$ is non-empty if and only if $M,q\models \lnot \knows_a \lnot\alpha_i$.
    Therefore $M,q\models \uncertain_a^{=\log m}\beta$ if and only if for some subsets $\varPhi_\beta'\subseteq \varPhi_\beta$ such that $|\varPhi_\beta'| = m$
    , we have that $\alpha\in \varPhi_\beta' \Rightarrow M,q\models \lnot \knows_a \lnot\alpha_i$
    and $\alpha\in \varPhi_\beta\setminus\varPhi_\beta' \Rightarrow M,q\models \knows_a \lnot\alpha_i$.
    This statement could be exactly represented by the formula $P_m(\varphi_1, ..., \varphi_n)$ defined as above.
  \end{proof}

  For formulas of type $\uncertain_a^{<\log m}\beta$ and $\uncertain_a^{>\log m}\beta$ the construction could be written as below:
  \begin{center}
    $M,q\models \uncertain_a^{<\log m}\beta \ \Longleftrightarrow \ M,q\models_{\knows}P_1(\varphi_1,...,\varphi_n) \vee ... \vee P_{m-1}(\varphi_1,...,\varphi_n)$
    $M,q\models \uncertain_a^{>\log m}\beta \ \Longleftrightarrow \ M,q\models_{\knows}P_{m+1}(\varphi_1,...,\varphi_n) \vee ... \vee P_{2^n}(\varphi_1,...,\varphi_n)$
  \end{center}
  where $P(N)$ is defined as in Proposition~\ref{prop:expressivity}.

\end{proposition}

\subsection*{Proof of Theorem~\ref{theorem:succinctness}}

For introducing the proof of Theorem~\ref{theorem:succinctness} we need some prerequisites which will follow.

\begin{lemma}\label{lemma:simplify1}
    In any concurrent game model $M$, if the outcome function contains only reflexive transitions, then for any $A\subseteq\Agt$ and any state $q$:
    \begin{itemize}
        \item $M,q\models \coop{A}\Next\varphi$ iff $M,q\models E_A\varphi$.
        \item $M,q\models \coop{A}\Always\varphi$ iff $M,q\models E_A\varphi$.
        \item $M,q\models \coop{A}\psi\Until\varphi$ iff $M,q\models E_A\varphi$.
    \end{itemize}
\end{lemma}

\begin{proof}
    The proof follows directly from the semantics of the temporal operators and the fact that for any strategy $s_A\in\Sigma_A^{ir}$ and any $\lambda\in out(q,s_A)$,
    we have that $\lambda[i] = q$ for any $i\geq 0$. So in a sense the agents can never ''escape'' from the state $q$.
    As the indistinguishability relation of the agents are also reflexive, it follows that the strategic abilities of agents to bring about $\varphi$ in any future state is equivalent to having mutual knowledge about $\varphi$ in the current state.
\end{proof}

\begin{observation}\label{obs:simplify2}
    In any concurrent game model $M$, if the outcome function contains only reflexive transitions and $\Agt$ is a singleton (for example $\Agt=\set{a}$),
    then any formula in the language of \ATEL has an equivalent formula in Multiagent epistemic logic which does not have a higher length.
    This follows from the lemma~\ref{lemma:simplify1} and that in this model $E_A$ is equivalent to $\knows_a$.
\end{observation}

\begin{definition}\label{def:Mn}{\ \\}
    For $n\in\mathbb{N}$, we define $M^n = \set{\Agt^n, \States^n, \set{\indist_a \mid a\in\Agt}, \Act^n, d^n, \outcome^n, \atomicP^n, \valualtion^n}$ to be a concurrent epistemic game model,
    such that $\Agt^n = \set{a}$, $\States^n = \set{0,..., 2^n-1}$, $\atomicP^n = \set{p_1, ..., p_n}$ and $s\indist_a t, \forall s,t\in\States$.
    We allow $\Act^n$, $d^n$ to be defined arbitrarily, and the transition function $\outcome^n$ is reflexive.
    The valuation function $\valualtion^n$ is defined as follows: If $d_n d_{n-1} ...d_1$ is the binary representation
    of the number $t\leqslant 2^n-1$, then in state $t\in\States$, we let $p_i$ be true iff $d_{i-1} = 1$.
    In other words, we let the set of $n$ atomic propositions $\atomicP^n$ have a distinct valuation out of $2^n$ possible combinations in each of $2^n$ states.
  \end{definition}

  \begin{definition}\label{def:Nnj}{\ \\}
    For a given $n,j\in\mathbb{N}$, we define $N^n_j$ to be constructed by removing state $j$ from $M^n$.
    We also modify $\indist^n$, $\Act^n$, $d^n$, $\valualtion^n$ and $\outcome^n$ as needed to be compatible be the removing of the state.
  \end{definition}

  We denote the set of states of $N^n_j$ by $\States^n_j$.
  We define $\widetilde{N}^n_j$ to be the set of all pointed models $(N^n_j,s)$, where $s\in\States^n_j$.
  Therefore $\widetilde{N}^n_j$ contains $2^n-1$ pointed models, corresponding to each of $2^n-1$ states of $N^n_j$.
  For any $t\in \States^n$, we define $(M^n, t)^+$ to be any set of pointed models that contains $(M^n, t)$. Similarly $(N^n_j, t)^+$ could be any set
  of pointed models that contain $(N^n_j, t)$ (therefore, these terms are not referring to specific sets, but rather they show a property of a set of pointed models).

  \begin{observation}\label{obs:obs1}
    If $n>1$, then $\forall p_i\in\atomicP^n, \forall s\in\States^n$, there exists $t\in\States^n, t\neq s$ such that $s$ and $t$ have same valuation for $p_i$.
    This is true, because for $n>1$, $M^n$ has at least 4 states, and each $p_i$ holds in at least 2 of them, and does not hold in at least 2 other states.
  \end{observation}

  \begin{observation}\label{obs:obs2}
    If $(M^n,t)\in A$ and $(N^n_j,t)\in B$, then the spoiler cannot play an \textit{atomic move} in a node which is labeled by $\pairModels{A}{B}$ or $\pairModels{B}{A}$.
    Because state $t$ prevents any $p\in\atomicP^n$ to hold on all models on the left side and to not hold on all models on the right side.
  \end{observation}

  We will use the models $M^n$ and $N^n_j$ to construct our proof.
  Throughout the proof, we will create a formula size game tree.
  As we saw before, every node in a formula size game tree is labeled with a pair of sets of pointed models. Below we define six \textit{types} of nodes,
  parameterized by a given state $j\in \States^n$, based on properties of their labels.

  \begin{itemize}
    \item \sOne:  The label could be represented by $\pairModels{(M^n, t)^+}{(N^n_j, t)^+},\, \exists t\neq j$
          \vspace{6pt}
    \item \sTwo:  The label could be represented by $\pairModels{(N^n_j, t)^+}{(M^n, t)^+},\, \exists t\neq j$
          \vspace{6pt}
    \item \sThree:  The label could be represented by $\pairModels{\widetilde{N}^n_j}{(M_n, j)}$
          \vspace{6pt}
    \item \sFour:  The label could be represented by $\pairModels{(M^n, j)}{\widetilde{N}^n_j}$
          \vspace{6pt}
    \item \sFive:  The label could be represented by $\pairModels{(N^n_j, t)^+}{(M^n, j)},\, \exists t\neq j$
          \vspace{6pt}
    \item \sSix:  The label could be represented by $\pairModels{(M^n, j)}{(N^n_j, t)^+},\, \exists t\neq j$
  \end{itemize}

  Next, we investigate that if in a FSG, a node has one of the above types, what will be the result of each of the spoiler's possible moves.
  In particular, we are interested to see a) whether the spoiler can close a node of a certain type be playing one of her moves, and b) if it is not closed then, what type of child nodes will be generated.

  \vspace{5pt}
  \noindent\sOne, defined by  $\pairModels{(M^n, t)^+}{(N^n_j, t)^+},\, \exists t\neq j$\\
  - \highlightone{\textit{Atomic move} is not possible} (by observation~\ref{obs:obs2}).\\
  - \textit{Or move} will generate at least one child node of type \sOne. Because it splits the set of pointed models on the left side, therefore $(M^n,t)$ will end up
  on the left side of one of the child nodes.\\
  - \textit{Knows move} will generate a child node of type \sOne. Because it expands the pointed models on the left side to $\set{(M^n, s)|\, \forall s\in\States^n}$.
  Therefore, no matter how we generate the right side, we can always select one of the pointed models on the right side which are of form $(N^n_j, t)$ (for some $t\in\States^n_j$)
  and because $t\in\States^n$, therefore $(M^n,t)$ will be included on the left side. Hence the child node is of type \sOne.\\
  - \textit{Not move} will generate a child node of type \sTwo. Because it just reverses the pointed models on right side and left side.\\

  \noindent\sTwo, defined by $\pairModels{(N^n_j, t)^+}{(M^n, t)^+},\, \exists t\neq j$\\
  - \highlightone{\textit{Atomic move} is not possible} (by observation~\ref{obs:obs2}).\\
  - \textit{Or move} will generate at least one child node of type \sTwo. Because it splits the set of pointed models on the left side, therefore $(N^n_j,t)$ will end up
  on the left side of one of the child nodes.\\
  - \textit{Knows move} will generate a child node of either type \sTwo or \sThree. By playing the \textit{knows move} (which could be only $\knows_a$ move), the left side
  of the label in the generated node will be expanded to $\widetilde{N}^n_j$. On the right side, we either select pointed models such that at least for one $(M^n, t)$ among them, $(N^n_j,t)$ exists on the left side,
  in which case the resulting node will be of type \sTwo. Or we select the pointed models on the right side such that for no $(N^n_j,t)$ on the left side there exist
  a $(M^n,t)$ on the right side. This can happen only if the new generated node is labeled with $\pairModels{\widetilde{N}^n_j}{(M_n, j)}$, which means it is of type \sThree.\\
  - \textit{Not move} will generate a child node of type \sOne. Because it just reverses the pointed models on right side and left side.\\

  \noindent\sThree, defined by $\pairModels{\widetilde{N}^n_j}{(M_n, j)}$\\
  - \highlightone{\textit{Atomic move} is not possible}. Because by observation~\ref{obs:obs2}, for any $p\in\atomicP^n$ which is not true in $j$, there is $s\in\States^n_j$ where
  where $p$ does not hold in $s$ too, therefore no atomic proposition could close this node.\\
  - \textit{Or move} will generate at least one child node of either type \sThree or \sFive. If we select the all the pointed models on the left side, to be the left side of
  a child node, then the label of that child node is the same as this node and therefore it is also of type \sThree. Otherwise we can select any of the child nodes generated by \textit{or move}
  and one of the pointed models $(N^n_j, t)$ on its left side. Because $t\neq j$, this node's label will be of type \sFive.\\
  - \textit{Knows move} will generate a child node of either type \sTwo or \sThree. By playing the \textit{knows move} (which could be only $\knows_a$ move), the left side
  of the label in the generated node will remain to be $\widetilde{N}^n_j$. On the right side, we either select $(M^n,t), t\neq j$, in which case $(N^n_j,t)$ exists on the left side,
  and therefore the resulting node will be of type \sTwo. Or we select the pointed model $(M^n,j)$,
  in which case the label of the child node stays intact and hence is of type \sThree.\\
  - \textit{Not move} will generate a child node of type \sFour. Because it just reverses the pointed models on right side and left side.\\

  \noindent\sFour, defined by $\pairModels{(M^n, j)}{\widetilde{N}^n_j}$\\
  - \highlightone{\textit{Atomic move} is not possible}. Because by observation~\ref{obs:obs2}, for every $p\atomicP^n$, there is $s\in\States^n_j, s\neq j$
  such that $p$ does hold in both $s$ and $j$, therefore no atomic proposition could close this node.\\
  - \textit{Or move} will generate a node of type \sFour. This is because the set of pointed models on the left side of this node has only one member, so at least one of the child
  nodes will have a similar label as of this one, which is of type \sFour.\\
  - \textit{Knows move} will generate a child node of \sOne. Because after a \textit{knows move} (which could be only $\knows_a$ move), the left side
  of the label in the generated node will be $\set{(M^n, s)|\ \forall s\in \States^n}$. Therefore no matter how we generate the pairs on the right side, for any chosen pair $(N^n_j, t)$
  on the right side $(M^n, t)$ exists on the left side, and therefore the resulting node would be of type \sOne.\\
  - \textit{Not move} will generate a child node of type \sThree. Because it just reverses the pointed models on right side and left side.\\

  \noindent\sFive, defined by $\pairModels{(N^n_j, t)^+}{(M^n, j)},\, \exists t\neq j$\\
  - \highlighttwo{\textit{Atomic move} is possible}. If it is played, then the node will be closed and there won't be any further child. Notice that we are not saying the for every
  node of type \sFive the \textit{atomic move} is possible. What we state is that by the definition of this type alone, we cannot discard the possibility of an \textit{atomic move},
  and this will be enough for our purpose.\\
  - \textit{Or move} will generate a child node of type \sFive. Because for all $(N^n_j,t)$ on the left side of this node's label, we have that $t\neq j$. Therefore any subset of it also
  retains this property and hence an \textit{or move} will generate at least one child node of type \sFive.\\
  - \textit{Knows move} will generate at least one node of either type \sThree or \sTwo. The reasoning is similar to the \textit{knows move} for the node type \sThree.\\
  - \textit{Not move} will generate a child node of type \sSix. It follows directly from the definitions of node types \sFive and \sSix.\\

  \noindent\sSix, defined by $\pairModels{(M^n, j)}{(N^n_j, t)^+},\, \exists t\neq j$\\
  - \highlighttwo{\textit{Atomic move} is possible}. Again, this only means that only by the definition of type \sSix, we cannot reject the possibility of an \textit{atomic move},
  which will be enough for our proof.\\
  - \textit{Or move} will generate a child node of type \sFive. The reasoning is similar to the \textit{or move} of node type \sFour.\\
  - \textit{Knows move} will generate a child node of type \sOne. The reasoning is similar to the \textit{knows move} of node type \sFour.\\
  - \textit{Not move} will generate a child node of type \sSix. It follows directly from the definitions of node types \sFive and \sSix.\\

  Any node which is tagged with a type \sOne to \sSix defined as above is called a \textit{typed} node.

  \begin{observation}\label{obs:obs3}
    If the spoiler wins a FSG starting from a root of type \sOne,
    then the game tree has at least one leaf of either type \sFive or \sSix (with a singleton set containing $(M^n, j)$ on one side of the label of that leaf).
  \end{observation}

  To see why, notice that every move that the spoiler plays at each step will generate at least one typed node parameterized by $j$ as a child. So if we create
  a path from the root downwards the tree, containing only typed nodes parameterized by $j$, this path can end only by the node being closed at some point.
  As we assumed that the spoiler wins the FSG, then all paths from the root has to end in a leaf. Among the node types, the spoiler
  can close the node only when the type is \sFive or \sSix. Therefore, any path that contains only of typed nodes parameterized by $j$ has to end in a node of types either \sFive or \sSix.

  Now we have enough prerequisite materials to construct our proof:
  \begin{proof}{\textit{(Proof of Theorem~\ref{theorem:succinctness})}}\label{proof:succinctness}\\

    For every $n\in\mathbb{N}$, we define the formula $\varphi_n\in Lang(ATHL)$ as:
    \begin{center}
      $\varphi_n = \uncertain_a^{=n} \beta_n$
    \end{center}
    where $\beta_n = \set{p_1, ..., p_n}$. The size of $\varphi_n$ is $f(n) = n$, which is strictly increasing (and linear) in $n$.
    Also for every $n\in\mathbb{N}$, we define two set of pointed models $A_n$ and $B_n$ as below:

    \begin{center}
      $A_n = \set{(M^n,0)}$,
      $B_n = \set{(N^n_j, 0)\mid j\in \mathbb{N},\, 1\leqslant j \leqslant 2^n-1}$
    \end{center}

    where $M^n$ and $N^n_j$ are defined as in Definitions~\ref{def:Mn} and \ref{def:Nnj}. So $A_n$ contains a single pointed model $(M^n, 0)$,
    and $B_n$ is constructed by summing all possible ways that we can remove one of the states of $(M^n, 0)$, other than state 0. Therefore $B_n$ consists of $2^n-1$ pointed models.

    It holds that $A_n\models_{\uncertain}\varphi_n$, because $M^n$ has $2^n$ states and the set $\beta_n = \atomicP^n$ has a distinct combination of values
    in each of those $2^n$ states, which are all in the same indistinguishability class for agent $a$ from state $0$.
    Therefore $|R_{a,0}(\beta)| = 2^n$, and so $M,0\models_{\uncertain}\uncertain_a^{=n}\beta_n$.
    It is also true that $B_n\models_{\uncertain}\lnot\varphi_n$, because for all $j\in\mathbb{N}$, $N^n_j$ has $2^n-1$ states, and so regardless of the
    valuations of atomic propositions in those states, $|R_{a,0}(\beta)| \leqslant 2^n-1 < 2^n$. Hence $\log{|R_{a,0}(\beta)|}<\log{2^n} = n$. It follows that
    $(N^n_j, 0)\models_{\uncertain}\uncertain_a^{<n}$, which means $(N^n_j, 0)\models_{\uncertain}\lnot\uncertain_a^{=n}$.

    So we have $A_n\models_{\uncertain}\varphi_n$ and $B_n\models_{\uncertain}\lnot\varphi_n$. Now to prove the succinctness theorem, by definition~\ref{def:succinctness}
    we need to show that the shortest formula $\psi_n\in Lang(ATEL)$ which is equivalent to $\varphi_n$, has the size of at least $2^n$.
    Assume $\psi_n'$ is the shortest formula in $Lang(ATEL)$ which is equivalent to $\varphi_n$.
    By Observation~\ref{obs:simplify2}, there exist a formula $\psi_n$ in $Lang(MEL)$ which is not longer than $\psi_n'$.
    By theorem~\ref{theorem:formulaSizeGame}, the spoiler can win the FSG game starting at $\pairModels{A_n}{B_n}$. If $k$ is the lowest number of moves that the spoiler can win
    the FSG game starting at $\pairModels{A_n}{B_n}$, then by theorem~\ref{theorem:formulaSizeGame}, $k$ is a lower bound for the size of $\psi_n$.
    Now consider that the spoiler actually plays and wins this FSG game and constructs this \textit{winning} tree.
    We will \textit{tag} each node of this tree by its type, based on the types \sOne to \sSix that we defined before.
    We start from the root. For each  $1\leqslant j \leqslant 2^n-1$, we can \textit{tag} the root of this tree as type \sOne by replacing  $t = 0$ in the definition of type \sOne.
    Therefore, by observation~\ref{obs:obs3}, this tree has at least one leaf of either type \sFive or \sSix for each $1\leqslant j \leqslant 2^n-1$. Because we have $2^n-1$ choice for $j$,
    and because no node can be of types \sFive or \sSix, and at the same time of types $\mathbf{S5_i}$ or $\mathbf{S6_i}$ for $i\neq j$, therefore this tree has at least $2^n-1$ distinct leaves.
    Moreover, because the root is of type \sOne and cannot be closed with the first move, the tree has at least one no-root/no-leaf node. So to win this game the spoiler has to play
    at least $2^n$ moves, and because this is a lower bound on the size of $\psi_n$, we have that $|\psi_n|\geqslant 2^n$.
    As $\psi_n'$ can only be longer than $\psi$, it follows that $|\psi_n'|\geqslant 2^n$.

    So we showed that for any $n\in\mathbb{N}$, we can define a $\varphi_n\in Lang(ATHL)$ with length $n$, such that the shortest $\psi_n'\in Lang(ATEL)$ which is
    equivalent to $\varphi_n$ has at least the length of $2^n$. By definition~\ref{def:succinctness} this proves that \ATHL is exponentially more succinct than \ATEL.
  \end{proof}
 }

\end{document}